\documentclass[a4paper,UKenglish,cleveref, autoref, thm-restate, numberwithinsect]{lipics-v2019}
\usepackage{makecell}
\usepackage{tikz}
\usepackage{multirow}
\usepackage[textwidth=22mm,textsize=tiny,disable]{todonotes}
\usepackage{mathtools} % added for coloneqq
\usepackage[ruled,vlined]{algorithm2e} 

\newcommand{\setf}{\mathrm{set}}
\newcommand{\TO}{\tilde{O}}

\newcommand{\range}[2]{{\sc Range (#2, #1)}}

\newcommand{\rangeqint}[2]{{\sc Range (#2, #1)}}
\newcommand{\problemname}[2]{{\sc 2D Grid \rangeqint{#1}{#2}}}

\newcommand{\pointwise}[3]{{\sc #3D Range (#2, #1)}}
\newcommand{\setsum}{\problemname{$\setf$}{$+$}}

\newcommand{\ind}{1}
\newcommand{\numupdates}{\ensuremath{n_{\textnormal u}}}
\newcommand{\numqueries}{\ensuremath{n_{\textnormal q}}}
\newcommand{\jl}[1]{\todo[fancyline]{\textbf{JL:} #1}}

\newtheorem{conjecture}[theorem]{Conjecture}

\newtheorem{openproblem}[theorem]{Open Problem}

\title{Algorithms and Hardness for Multidimensional Range Updates and Queries}

\author{Joshua Lau}{Sydney, Australia}{joshua.cs.lau@gmail.com}{https://orcid.org/0000-0001-7490-633X}{}%mandatory, please use full name; only 1 author per \author macro; first two parameters are mandatory, other parameters can be empty. Please provide at least the name of the affiliation and the country. The full address is optional

\author{Angus Ritossa}{UNSW Sydney, Australia}{a.ritossa@unsw.edu.au}{https://orcid.org/0000-0002-9807-773X}{}

\authorrunning{J. Lau and A. Ritossa} %mandatory. First: Use abbreviated first/middle names. Second (only in severe cases): Use first author plus 'et al.'

\Copyright{Joshua Lau and Angus Ritossa} %mandatory, please use full first names. LIPIcs license is "CC-BY";  http://creativecommons.org/licenses/by/3.0/


 %mandatory; please add comma-separated list of keywords

\relatedversion{} %TODO: optional, e.g. full version hosted on arXiv, HAL, or other respository/website
\acknowledgements{We thank Tunan Shi, for suggesting the reduction from {\sc 2RangeInversionsQuery} to {\sc Static Trellised \problemname{$\setf$}{$+$}}. We also thank Ray Li and anonymous reviewers for helpful suggestions and comments.}

\nolinenumbers %uncomment to disable line numbering

\hideLIPIcs  %uncomment to remove references to LIPIcs series (logo, DOI, ...), e.g. when preparing a pre-final version to be uploaded to arXiv or another public repository

\EventEditors{James R. Lee}
\EventNoEds{1}
\EventLongTitle{12th Innovations in Theoretical Computer Science Conference (ITCS 2021)}
\EventShortTitle{ITCS 2021}
\EventAcronym{ITCS}
\EventYear{2021}
\EventDate{January 6--8, 2021}
\EventLocation{Virtual Conference}
\EventLogo{}
\SeriesVolume{185}
\ArticleNo{37}
\begin{document}

\maketitle

\begin{abstract}
    Traditional orthogonal range problems allow queries over a static set of points, each with some value.
    Dynamic variants allow points to be added or removed, one at a time.
    To support more powerful updates, we introduce the {\sc Grid Range} class of data structure problems over arbitrarily large integer arrays in one or more dimensions.
    These problems allow range updates (such as filling all points in a range with a constant) and queries (such as finding the sum or maximum of values in a range). %, resembling simple spreadsheet operations.
    In this work, we consider these operations along with updates that replace each point in a range with the minimum, maximum, or sum of its existing value, and a constant.
    In one dimension, it is known that segment trees can be leveraged to facilitate any $n$ of these operations in $\TO(n)$ time overall.
    Other than a few specific cases, until now, higher dimensional variants have been largely unexplored.

    Despite their tightly-knit complexity in one dimension, we show that variants induced by \emph{subsets} of these operations exhibit polynomial separation in two dimensions.
    In particular, no truly subquadratic time algorithm can support certain pairs of these updates simultaneously without falsifying several popular conjectures.
    On the positive side, we show that truly subquadratic algorithms can be obtained for variants induced by other subsets.

    We provide two general approaches to designing such algorithms that can be generalised to online and higher dimensional settings.
    First, we give almost-tight $\TO(n^{3/2})$ time algorithms for single-update variants where the update and query operations meet a set of natural conditions.
    Second, for other variants, we provide a general framework for reducing to instances with a special geometry.
    Using this, we show that $O(m^{3/2-\epsilon})$ time algorithms for counting paths and walks of length 2 and 3 between vertex pairs in sparse graphs imply truly subquadratic data structures for certain variants; to this end, we give an $\TO(m^{(4\omega-1)/(2\omega+1)}) = O(m^{1.478})$ time algorithm for counting simple 3-paths between vertex pairs.
\end{abstract}

\newcommand{\probs}{\ensuremath{\mathfrak{B}}}

\section{Introduction}

Orthogonal range query problems are ubiquitous across various fields of Computer Science.
In the simplest of these problems, a data set is modelled as a set of points in $\mathbb{Z}^d$, and the task is to design a data structure that can efficiently answer queries which ask: how many points lie within the (axis-aligned) orthogonal range $[l_1, r_1] \times \ldots \times [l_d, r_d]$?
One can extend this definition by assigning a value to each point in the input, and having queries ask instead for some aggregate (e.g. the maximum value) of the values of the points within the query range.
This has been studied extensively
\cite{Bentley1975, Chazelle1988, Agarwal2004, Chan2011, Farzan2012, He2014, Okajima2015, Chan2017a, Afshani2019, Nekrich2020}, along with models which ask to report all points in the range
\cite{Afshani2009, Afshani2010, Nekrich2020}.
\emph{Dynamic} models, where a single point may be inserted or removed in a single operation, have also been studied
\cite{Willard1985, Chazelle1988, Poon2003, Chan2017, Chan2019}, corresponding to the addition or deletion of a single record.
Queries may then be interspersed between these update operations, providing insight into the data set as it changes over time.
Modelling data sets in this way has proven useful for Online Analytical Processing (OLAP) \cite{Gray1996, Poon2003} of databases.

In practice, however, one may wish to employ more powerful updates.
In this work, we examine data structures which support updating the values of all points that fall within a range in addition to range queries.
This models updates to the records in a database table whose numerical field values fall within designated ranges.
For instance, this could be giving all employees who have been employed between 5 to 10 years, and have KPIs between 80 and 90 an ``A'' rating. %\jl{Add more applications}
We formalise this as follows.

\begin{definition}
    Let $d$ be a positive integer constant, $(\mathbb{Z}, q)$ be a commutative semigroup\footnote{A \emph{semigroup} $(S, +)$ is a set $S$ equipped with an associative binary operator $+ : S \times S \rightarrow S$. A semigroup is \emph{commutative} if $x + y = y + x$ for all $x, y \in S$.}
    and $U = \{u_j\}$ be a set of integer functions.
    The {\sc ($d$D) Range $(q, U)$} problem gives as input a set $P = \{x_1, \ldots, x_{p}\}$ of $p$ points in $\mathbb{Z}^d$.
    A corresponding integer \emph{value} $v_i$ is also given for each point $x_i$, each initially 0.
    It requests a series of \emph{operations}, each taking one of the following forms:
    \begin{itemize}
        \item $\text{update}_j((l_1, r_1), \ldots, (l_d, r_d))$: for each $x_i \in [l_1, r_1] \times \ldots \times [l_d, r_d]$, set $v_i := u_j(v_i)$
        \item $\text{query}((l_1, r_1), \ldots, (l_d, r_d))$: compute and return $q(P')$, where $P'$ is the multiset $\{v_i : x_i \in [l_1, r_1] \times \ldots \times [l_d, r_d]\}$.
    \end{itemize}
    Importantly, the operations may include updates of different types, and operations may occur in any order.
    All operations are given \textbf{online} and no information is known about them before the preceding operations are performed.
    There are $n = \numupdates + \numqueries$ operations in all, with $\numupdates$ and $\numqueries$ of the first and second forms, respectively.
\end{definition}

\noindent
We are most interested in the case where update functions take the form $\ast_c(x) = x \ast c$, where $\ast$ is a binary operation over the integers, and $c$ is an operation-specific constant.
Through a slight abuse of notation, we also use $\ast$ to denote the set of all functions of the form $\ast_c$.
We write $\ast$ as a member of $U$ as shorthand for $\ast_c$ being a member of $U$, for all $c$.
For example, {\sc Range $(+, \{+, \max\})$} allows updates of the form $+_c$ (increasing values by $c$), $\max_c$ (replacing values less than $c$ with $c$), and queries which ask for the sum of values in a range.
For a single update function or binary operation $u$, we write {\sc Range $(q, u)$} for short.

The {\sc Grid Range} variants are those whose point set is $P = [s]^d$; $P$ is not given explicitly, but rather, is described in the input by the positive integer $s$.
Hence, the size of input (and thus the running time) can be measured as a function of the number $n$ of operations which occur, rather than the size of $P$.
{\sc Grid Range} problems will form our primary focus, but we first describe the context surrounding {\sc Range} problems as a whole.

\subsection{Motivation}

In one dimension, {\sc Range} problems can be viewed as operating over an array.
In this case, balanced binary trees (such as binary search trees or \emph{segment trees} \cite{Bentley1977}) over the array have been effective tools in solving {\sc Range} problems.
Such trees allow any range of the array to be canonically expressed as the disjoint union of $O(\log p)$ subtrees of the tree.

When the functions in $U$ are closed under composition and each distributes\footnote{An integer function $u$ distributes over $q$ if $u(q(a, b)) = q(u(a), u(b))$} over $q$, the folklore technique of \emph{lazy propagation} over a binary tree solves {\sc 1D \range{$U$}{$q$}} in $O(\log p)$ time in the worst case, per operation.
Lazy propagation can be extended to support several types of updates: it can be applied to solve {\sc 1D \range{$\{+, \min, \setf, \max\}$}{$\max$}} in worst-case $O(\log p)$ time per operation, where $\setf$ is the binary operation that returns its second operand.
These techniques are described in more detail in \autoref{sec:prelim}.
Ji \cite{Ji2016} considered the {\sc 1D \range{$\max$}{$+$}} problem, where the update operation does not distribute over the query operation, by introducing a technique known as a Ji-Driver Tree (informally ``Segment Tree Beats''), generalising lazy propagation.
Using this technique, he showed that {\sc 1D \range{$\{+, \min, \setf, \max\}$}{$+$}} can be solved in amortised $O(\log^2 p)$ time per operation.

Let \probs{} be the set of {\sc \range{$U$}{$q$}} problems with $q \in \{+, \max\}$ and $U \subseteq \{+, \min,$ $\setf, \max\}$.
Motivated by the one dimensional results for problems in \probs, we seek general techniques addressing {\sc Range} problems in two or more dimensions.
This has been asked as an open question in the competitive programming community \cite{Korhonen2019}.

The {\sc Range} problem on a $p$ element array can be generalised to multiple dimensions in two natural ways: either a set of $p$ points in $\mathbb{Z}^d$ is provided explicitly, or the point set is considered to be $[s]^d$.
In the former case, one dimensional techniques can be generalised to higher dimensions with the aid of an orthogonal space partition tree (hereafter simply \emph{partition tree}), such as a $k$d-tree \cite{Bentley1975}.

\begin{restatable}{lemma}{kdtree}
    \label{kdtree}
    If {\sc 1D \range{$U$}{$q$}} on $p$ points can be solved in time $T(p)$ per operation and $q$ is computable in $O(1)$ time, then \pointwise{$U$}{$q$}{$d$} can be solved in time $O(p^{1-1/d} T(p))$ per operation.
\end{restatable}

\noindent
We prove this result and provide almost-tight $\Omega(p^{1/2-o(1)})$ time per-operation lower bounds conditioned on the Online-Matrix Vector (OMv) Conjecture of Henzinger et al. \cite{Henzinger2015}, for all {\sc Range} problems in \probs{} when $d = 2$, in \autoref{sec:pointwise}.

The latter case corresponds to the {\sc Grid Range} class of problems, which is the subject of the remainder of this work.
The same technique does not apply in these cases, as the number of points is $p = s^d$.

\subsection{Prior work}

Prior work on {\sc Grid Range} problems has been limited to a few specific cases.

Since balanced binary trees have been rather effective in solving problems in one dimension, it is natural to ask whether they can be easily generalised to higher dimensions.
Lueker \cite{Lueker1978} and Chazelle \cite{Chazelle1988} generalised binary search trees and segment trees to higher dimensions to answer various $d$-dimensional range query problems (without updates) on sets of $n$ points in $O(\log^d n)$ time per query.
This technique can be used to solve {\sc Grid Range} problems in the special cases where all update ranges affect a single point, or when all query ranges contain a single point and the update functions are commutative.
The latter can be seen as a generalisation of {\sc Rectangle Stabbing} \cite{Chan2019} (which accepts a series of $d$-dimensional boxes and supports queries for the number of boxes covering a given point), a dual of traditional range queries.
The same structure was used with scaling by Ibtehaz et al. \cite{Ibtehaz2018} to solve {\sc Grid \range{$+$}{$+$}} in worst-case $O(\log^d s)$ time per operation; without scaling, a straightforward solution for {\sc Grid \range{$\max$}{$\max$}} can be obtained in the same time.
These upper bounds contrast with the abovementioned $\Omega(p^{1/2-o(1)})$ time conditional lower bounds for the corresponding {\sc 2D Range} problems.
For other problems in \probs, the lazy propagation technique used in one dimension cannot be directly applied over this structure, as it requires a partition tree of the coordinate space.

In the {\sc Klee's Measure} problem, $n$ rectangles are given in $d$ dimensions, and one is asked to find the volume of their union.
In {\sc Weighted Depth}, the rectangles each have a weight, and one is asked to find the maximum sum of weights of rectangles that cover a single point.
In {\sc Dynamic} versions of these problems, a single rectangle can be added or removed in a single operation, and the new answer must be returned after each one.
These can be seen as special cases of {\sc Grid Range} problems.

When only additions are supported, {\sc Dynamic Klee's Measure} is a special case of {\sc Grid \range{$\setf$}{$+$}} or {\sc Grid \range{$\max$}{$+$}}.
Overmars and Yap \cite{Overmars1991} gave a $O(n^{(d+1)/2}\log n)$ time solution over $n$ updates when \emph{the rectangles' coordinates are known during preprocessing}. %\ar{should it be (d+1)/2 instead of d/2?}
Chan \cite{Chan2008} gave a sublogarithmic time improvement over this method and showed that this is nearly tight in two dimensions, giving an $\Omega(\sqrt{n})$ time per update worst-case lower bound by reducing from {\sc Dynamic Matrix-Vector Multiplication}.

{\sc Dynamic Weighted Depth} is the special case of {\sc Grid \range{$+$}{$\max$}}, where the maximum value in the entire grid is queried after each update.
It is closely related to {\sc Dynamic Klee's Measure}, in that every known algorithm for one has been adaptable to the other, with running times differing only by a sublogarithmic factor.
Hence, with a slight modification of the result of Overmars and Yap \cite{Overmars1991}, an $\TO(n^{(d+1)/2})$ time algorithm for {\sc Grid \range{$+$}{$\max$}} can be obtained.
Chan \cite{Chan2008} showed that {\sc Dynamic Weighted Depth} is solvable in time $O(n^{(d+1)/2} \log^{d/2} \log n)$, with sublogarithmic improvements specific to entire-grid queries.

When there are no updates, {\sc Klee's Measure} and {\sc Weighted Depth} can be reduced to $O(n)$ updates on a $d-1$ dimensional {\sc Dynamic} instance with a sweepline, however Chan \cite{Chan2013} gave faster $O(n^{d/2-o(1)})$ time algorithms for these.
Reductions by Chan \cite{Chan2008} and Backurs et al. \cite{Backurs2016} showed that solving static variants of {\sc Klee's Measure} or {\sc Weighted Depth} in time $o(n^{d\omega/6})$ or $O(n^{d/2-\epsilon})$ for some $\epsilon > 0$, respectively, would improve longstanding upper bounds for {\sc Clique} or {\sc Max Clique}, respectively.
It follows that these problems are W[1]-complete for parameter $d$.
The same reductions can be appropriated for $d = 3$ to show that an $O(n^{3/2-\epsilon})$ time algorithm for {\sc 3D Weighted Depth} implies an $O(n^{3-2\epsilon})$ time algorithm for {\sc Negative Triangle}.
A truly subcubic algorithm for {\sc Negative Triangle} exists if and only if one exists for {\sc APSP} \cite{Williams2010}, so an $O(n^{3/2-\epsilon})$ time algorithm for \problemname{$+$}{$\max$} for some $\epsilon > 0$ would falsify the APSP Conjecture.

\subsection{Our contribution}

Our main contributions are conditional lower bounds, which serve to elucidate and categorise the expressiveness of these data structures, and upper bounds in the form of general algorithms and frameworks for solving {\sc Grid Range} problems (such as those in \probs) in two or more dimensions.
We do so on a Word RAM with $\Omega(\log n)$-bit words.
In the sequel, we use $\probs'$ to denote \probs{} without the problems {\sc Grid \rangeqint{$+$}{$+$}} and {\sc Grid \rangeqint{$\max$}{$\max$}}, for which per-operation polylogarithmic time upper bounds are already known.

\subparagraph*{Lower bounds.}
First, we consider algorithms whose per-operation time complexity is a function of $s$, the side length of the grid.
In two dimensions, we show that there is no algorithm running in time $O(s^{1-\epsilon})$ per-operation for any $\epsilon > 0$, for any problem in $\probs'$, unless the OMv Conjecture is false. % \jl{Maybe generalise to semirings?}
Further, for {\sc Grid \rangeqint{$\{+, \max\}$}{$+$}} we obtain identical lower bounds, conditioned on the ``extremely popular conjecture'' of Abboud et al. \cite{Abboud2015} that at least one of the APSP, 3SUM and 2-OV (Orthogonal Vectors) Conjectures are true (the latter of which is implied by the Strong Exponential Time Hypothesis \cite{Williams2005}).
Hence, we cannot solve this variant in two dimensions in $O(s^{1-\epsilon})$ time per operation for any $\epsilon > 0$ without making a powerful breakthrough across several fields of Theoretical Computer Science.
For {\sc Grid \rangeqint{$\{+, \min\}$}{$\max$}} and {\sc Grid \rangeqint{$\{+, \max\}$}{$+$}}, we generalise our results to $d$ dimensions under the $d$-OV Conjecture to obtain an $\Omega(s^{(d-1)-o(1)})$ time per-operation lower bound.
All these lower bounds are almost-tight, as $\TO(s^{d-1})$ time per-operation upper bounds can be obtained by maintaining $s^{d-1}$ one dimensional instances.

These results, however, do not preclude the existence of efficient and practical algorithms for {\sc Grid Range} problems whose overall complexity is a function of the number of operations, $n$ -- the true size of the input -- rather than $s$.
Our aforementioned lower bounds under the OMv and APSP Conjectures translate to conditional $\Omega(n^{3/2-o(1)})$ lower bounds for all problems in $\probs'$, but our lower bound under the 3SUM Conjecture translates to conditional lower bounds of $\Omega(n^{2-o(1)})$ for \problemname{$\{+, \max\}$}{$+$}. %\ar{and {\sc Grid \rangeqint{$\{+, \min\}$}{$\max$}}}.
In $d$ dimensions, our lower bounds under the $d$-OV Conjecture translate to $\Omega(n^{d-o(1)})$ time conditional lower bounds for {\sc Grid \rangeqint{$\{+, \min\}$}{$\max$}} and {\sc Grid \rangeqint{$\{+, \max\}$}{$+$}}.
Our lower bounds are summarised in \autoref{tab:results} and proven in \autoref{sec:lowerbounds}.

\begin{table}%[ht]
    \begin{center}
        \def\arraystretch{1.5}
        \caption{Lower and upper bounds for {\sc Grid Range} problems in \probs, exhibiting polynomial separation. All results are in two dimensions where $d$ is unspecified, and in $d$ dimensions otherwise, where $d$ is a constant. All lower bounds hold offline, except those for OMv, and all upper bounds hold in fully online settings.}
        \resizebox{\textwidth}{!}{
        \begin{tabular}{ | c | c | c | c | }
            \hline
            $q$ & $U$ & {\bf Lower bounds} & {\bf Upper bound} \\
            \hline
            $\max$ & $\max$ & & $O(n \log^d s)$ (extension of segment trees) \\
            \hline
            $+$ & $+$ & & $O(n \log^d s)$ \cite{Ibtehaz2018} \\
            \hline
            \hline
        $\max$ & $+$ & \makecell[c]{$\Omega(n^{3/2-o(1)})$ APSP \cite{Backurs2016} or OMv \\ $\Omega(n^{(d+1)/2 - o(1)})$ {\sc MaxClique} \cite{Backurs2016} } & \makecell[c]{$\TO(n^{3/2})$ and $O(n\log^2n)$ space \\ $\TO(n^{(d+1)/2})$ \cite{Chan2008, Chan2013}} \\
            \hline
            $\max$ & $\min$ & \multirow{5}{*}{ $\Omega(n^{3/2-o(1)})$ OMv} & \multirow{2}{*}{$\TO(n^{(d+1)/2})$} \\ \cline{1-2}
            $\max$ & $\setf$ & &  \\ \cline{1-2}\cline{4-4}
            $\max$ & $\{\min, \setf, \max\}$ & & $\TO(n^{(d^2+2d-1)/2d})$ \\ \cline{1-2}\cline{4-4}
            $+$ & $\setf$ & & $\TO(n^{5/4 + \omega/(\omega+1)}) = O(n^{1.954})$  \\ \cline{1-2}\cline{4-4}
            $+$ & $\{+, \setf\}$ & & $\TO(n^{5/4 + (4\omega-1)/(4\omega+2)}) = O(n^{1.989})$  \\
            \hline
            \hline
            $\max$ & $\{+, \min\}$ & \multirow{2}{*}{$\Omega(n^{d-o(1)})$ $d$-OV} & \multirow{4}{*}{$\TO(n^d)$}  \\ \cline{1-2}
            $\max$ & $\{+, \min, \setf, \max\}$ & &  \\ \cline{1-3}
            $+$ & $\{+, \max\}$ & \multirowcell{2}{$\Omega(n^{2-o(1)})$ 3SUM\\ $\Omega(n^{d-o(1)})$ $d$-OV} & \\ \cline{1-2}
            $+$ & $\{+, \min, \setf, \max\}$ & &  \\
            \hline
        \end{tabular}
        }
        \label{tab:results}
    \end{center}
\end{table}

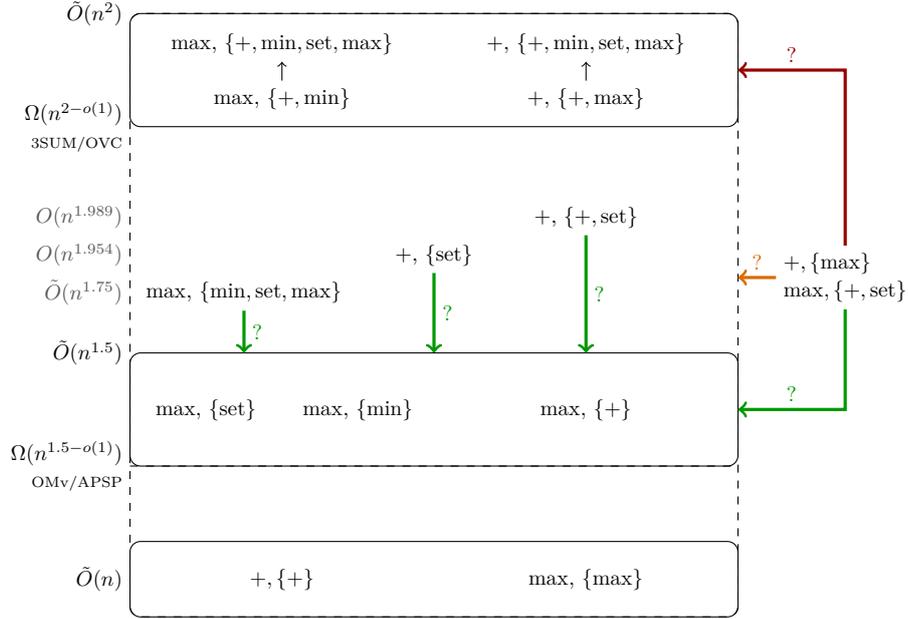
\begin{figure}
    \begin{center}
        \begin{tikzpicture}[every node/.style={scale=0.8}]
            \def\n{8}
            \draw[rounded corners, dashed] (0,0) rectangle (\n, \n);
            \node[left] at (0, \n) {$\tilde{O}(n^2)$};

            \def\l{0}
            \def\r{1}
            \draw[rounded corners] (0,\l) rectangle (\n, \r);
            \path (0, \l) -- (0, \r) node[left, midway] {$\tilde{O}(n)$};
            \node [align=center] at (2, 0.5) {$$+$, \{+\}$};
            \node [align=center] at (6, 0.5) {$\max$, $\{\max\}$};

            \def\l{6.5}
            \def\r{8}
            \node[align=center] (allmax) at (2, \n-.4) {$\max$, $\{+, \min, \setf, \max\}$};
            \node[align=center] (allplus) at (6, \n-.4) {$+$, $\{+, \min, \setf, \max\}$};
            \node[align=center, below=0.25cm of allplus] (plusminplus) {$+$, $\{+, \max\}$};
            \draw [->] (plusminplus) -- (allplus);
            \draw[rounded corners] (0,\l) rectangle (\n, \r);
            \path (0, \l) -- (0, \r) node[left, at start, align=right] {$\Omega(n^{2-o(1)})$\\{\scriptsize 3SUM/OVC}};
            \path (\n, \l) -- (\n, \r) coordinate[midway] (toprighthalf);
            \node[align=center, below=0.25cm of allmax] (plusminmax) {$\max$, $\{+, \min\}$};
            \draw [->] (plusminmax) -- (allmax);

            \def\l{2}
            \def\r{3.5}
            \def\offset{.3}
            \node [align=center] (plusmax) at (6, \l+.75) {$\max$, $\{+\}$};
            \node [align=center] (setmax) at (1, \l+.75) {$\max$, $\{\setf\}$};
            \node [align=center] (minmax) at (3, \l+.75) {$\max$, $\{\min\}$};
            \node [align=center] (setplus) at (4, \l+2.5+\offset) {$+$, $\{\setf\}$};
            \def\lowbar{1}
            \draw [dashed] (0, \l) -- (\n, \l) node[left, at start, align=right] {$\Omega(n^{1.5-o(1)})$\\{\scriptsize OMv/APSP}};
            \def\highbar{6.5}
            \draw [->, very thick, color=black!40!green] (setplus) -- (4, \r) node[midway, right]{?};
            \draw[rounded corners] (0,\l) rectangle (\n, \r);
            \path (0, \l) -- (0, \r) node[left, at end] {$\tilde{O}(n^{1.5})$};

            \path (\n, \l) -- (\n, \r) coordinate[midway] (midrighthalf);
            \node[left] at (0, \l+2.5+\offset) {\textcolor{black!30!gray}{$O(n^{1.954})$}};
            \node (plussetmax) at (6, \l+3+\offset) {$+$, $\{+, \setf\}$};
            \node[left] at (0, \l+3+\offset) {\textcolor{black!30!gray}{$O(n^{1.989})$}};
            \draw [->, very thick, color=black!40!green] (plussetmax) -- ++(0, -1.5-\offset) node[midway, right]{?};
            \node (minsetmaxmax) at (1.5, \l+2+\offset) {$\max$, $\{\min, \setf, \max\}$};
            \node[left] at (0, \l+2+\offset) {\textcolor{black!30!gray}{$\tilde{O}(n^{1.75})$}};
            \draw [->, very thick, color=black!40!green] (minsetmaxmax) -- ++(0, -.5-\offset) node[midway, right]{?};

            \path (\n, \l+1.5) -- (\n, \l+3.5) node[midway, right=.5, align=left] (unknowns) {$+, \{\max\}$\\$\max, \{+, \setf\}$} coordinate[midway] (unknownrighthalf);
            \draw [->, very thick, color=black!40!red] (unknowns) |- (toprighthalf) node[near end, above]{?};
            \draw [->, very thick, color=black!15!orange] (unknowns) -- (unknownrighthalf) node[midway, above]{?};
            \draw [->, very thick, color=black!40!green] (unknowns) |- (midrighthalf) node[near end, above]{?};

        \end{tikzpicture}
    \end{center}

    \caption{Our results for {\sc 2D Grid Range}. Questions left open that ask whether certain problems belong to certain complexity classes are marked with a `?'. See \autoref{sec:open} for more details on open problems.}

\end{figure}

\subparagraph*{Upper bounds.}
By reducing to algorithms in one dimension, $\TO(n^d)$ time algorithms can be found for all these problems.
Hence, we aim to determine which problems in $\probs'$ can be solved more efficiently, by seeking truly subquadratic time algorithms in two dimensions. %\ar{mention that subquadratic is for 2D?}.
To this end, we provide two general frameworks for developing such algorithms, both based on the approach of Overmars and Yap \cite{Overmars1991} for {\sc Dynamic Klee's Measure}.
Their algorithm constructs a partition tree of the grid such that the rectangles intersecting each leaf region form a ``trellis'' pattern. % \jl{Consider adding a diagram for trellised.}
This requires the coordinates of all rectangles to be known during precomputation.

First, we provide a fully-online generalisation of this approach, that does not require coordinates to be known ahead of time.%, obtaining the following result.

\begin{restatable}{theorem}{conditionsoy}
    \label{conditions-for-overmars-yap}
    Suppose $q$ and $u$ are associative, commutative binary operations, computable in $O(1)$ time, such that $u$ distributes over $q$, and 0 is an identity of $u$.
    Then {\sc Grid \range{$u$}{$q$}} can be solved in $\TO(n^{(d+1)/2})$ time.
\end{restatable}

This algorithm does not apply to problems such as the \rangeqint{$\setf$}{$+$} problem described in the abstract, as $\setf$ does not distribute over $+$.
Motivated by this, in \autoref{sec:minmax} we show that in two dimensions, efficient solutions to static {\sc Grid} \range{$U$}{$q$} instances where the update ranges form the same ``trellis'' pattern can be used to give fully-online, truly subquadratic time solutions to many {\sc Grid} \range{$U$}{$q$} problems.
We also extend these results to multiple dimensions.%, giving a detailed proof in the full version.

As an application, we use this approach to give a truly subquadratic time algorithm for %three problems in $\probs'$, in two dimensions.
\problemname{$\{\min, \setf, \max\}$}{$\max$}.
We do the same for \problemname{$\setf$}{$+$} and \problemname{$\{+, \setf\}$}{$+$} in \autoref{sec:setplus}, by drawing an equivalence and a reduction between the respective ``{\sc Static Trellised}'' instances and counting the number of 2- and 3-edge paths between vertex pairs, respectively.
To this end, we prove the following result.

\begin{restatable}{theorem}{pathquery}
    Let $G$ be a graph with $m$ edges and $O(m)$ vertices.
    The number of 3-edge walks between each of $q$ vertex pairs in $G$ can be found in $O(m^{2\omega/(2\omega+1)} (m+q)^{(2\omega-1)/(2\omega+1)})$ time.
\end{restatable}

In this way, queries on static graphs yield efficient, fully-online dynamic {\sc Grid Range} data structures.
We find it somewhat surprising that, though all of the problems in \probs{} can be solved in $\TO(n)$ time in one dimension, our upper and lower bounds imply likely polynomial separation in two or more dimensions (see \autoref{tab:results}).
% Specifically, between these three problems and those problems in $\probs'$ with a quadratic lower bound, in two dimensions .

Lastly, we provide a fully-online algorithm for \problemname{$+$}{$\max$} that uses $O(n\log^2 n)$ space, and runs in a time comparable to that of existing algorithms.

\begin{restatable}{theorem}{oldincmax}
    \problemname{$+$}{$\max$} can be solved in $\TO(n^{3/2})$ time, and $O(n \log^2 n)$ space.
\end{restatable}

This is proven in \autoref{sec:succinctplusmax}.
While our complexity is slower than existing results by Chan \cite{Chan2008} by a polylogarithmic factor, those require $O(n^{3/2+o(1)})$ space and are not fully-online.
Overmars and Yap \cite{Overmars1991} also gave an $O(n)$ space algorithm for static {\sc Klee's Measure} in $d$ dimensions using a sweepline, but this does not apply in the dynamic case.

% In the remaining sections, due to space constraints, we omit some proofs and sketch others.
% We refer the reader to the full version for full details and proofs.

\section{Preliminaries}
\label{sec:prelim}

\subparagraph*{Model of computation.}
All results are described are for a Word RAM over $l$-bit words, with $l = \Omega(\log n)$.
We further assume that any coordinates or values given in the inputs can be represented in a constant number of words, and that basic arithmetic and the (binary) operations used in range problems can be performed on a constant number of words in constant time.
In particular, $s = O(n^c)$, for some constant $c$.

\subparagraph*{Notation.}
We use the notation $\TO(f(n)) = O(f(n) \text{poly} \log n)$ to hide polylogarithmic factors.
Note that $\log^c s = \log^c n$ for any constant $c$.
Where our algorithms and proofs use positive or negative $\infty$ as a value, this can be replaced with a suitably large value, for a given input instance.

Where $x \leq y$ are real numbers, we denote by $[x, y]$ the set of all \emph{integers} between $x$ and $y$, inclusive.
When $y \geq 1$, we write $[y]$ as shorthand for $[1, y]$.

The binary operation $\setf$ is the operation whose value is its second operand.
That is, $\setf(a, b) = b$.

$2 \leq \omega < 2.37286$ \cite{Alman2020} is the exponent of multiplying two $n \times n$ integer matrices.
We also write $\omega(a, b, c)$ for the time taken to multiply an $a \times b$ matrix by a $b \times c$ matrix.

\subparagraph*{Ancillary problems and variants.}
In {\sc Range} problems, we say that the $i$th operation (update or query) occurs at \emph{time} $i$.
In {\sc Offline} variants, all operations are provided together with the initial input, and in {\sc Static} variants, it is guaranteed that all updates precede all queries.

We formalise and appropriate the ``trellis'' pattern observed by Overmars and Yap \cite{Overmars1991} for our use, as follows.
Call a {\sc Grid Range} instance {\sc Trellised} if for each update, there is a dimension $d^{\ast}$ such that  $[l_{d'}, r_{d'}] = (-\infty, \infty)$ for all $d' \in [d] \setminus \{d^{\ast}\}$.
When $d = 2$, updates must either cover all points in a range of rows, or all points in a range of columns, which we call \emph{row updates} and \emph{column updates}, respectively.

\subparagraph*{Segment trees and lazy propagation.}
Let $s$ be a power of two.
A \emph{segment tree} over an array $A$ containing $s$ elements is a complete rooted binary tree of ranges over $[s]$.
The root is $[1, s]$, and each node $[a, b]$ has two children: $[a, h]$, $[h+1, b]$, where $h = (a+b-1)/2$.
Hence, there are $O(s)$ nodes in the tree, with a depth of $\log s$.
Given an integer interval $I = [l, r] (1 \leq l \leq r \leq s)$, we can write $I$ as a canonical disjoint union of a set $\text{base}(I)$ of $O(\log s)$ nodes.
These are defined as the nodes closest to the root that are fully contained in $I$, and can be found recursively.
We say that an integer interval $[l_1, l_2]$ \emph{decomposes} to $I$ if $I \in \text{base}([l_1, l_2])$.

Segment trees can be used to prove the following folklore proposition.
\begin{proposition}[Lazy propagation]
    \label{lazyprop}
    Suppose $U$ is a set of update functions, and $q$ is a query function, computable in $O(1)$ time.
    If there is a set $\bar{U}$ such that:
    \begin{enumerate}
        \item $U \subseteq \bar{U}$ are sets of functions that can be represented and composed in $\TO(1)$ space and time, such that the composition of any series of at most $n$ (possibly non-distinct) functions of $U$ results in a function in $\bar{U}$; and
        \item For each $u \in \bar{U}$, $u$ distributes over $q$
    \end{enumerate}
    then {\sc 1D Grid \range{$U$}{$q$}} is solvable in $\TO(n)$ time.
\end{proposition}

\begin{proof}
    Without loss of generality, increase $s$ so it is a power of two, and build a segment tree $T$ over $[s]$.
    In each segment tree node $a$, we will store an element $u_a$ of $\bar{U}$.
We also store a ``$q$-value'' at each node $q_a$, each initially the value of $q$ over an array of $|a|$ 0s.
    If $a$ has children $b$ and $c$, we keep the invariant $q_a = u_a(q(q_b, q_c))$.
    Whenever we need to access node $a$, for each ancestor $a'$ from the root, we update $u_{b'} := u_{a'} \circ u_{b'}$ , $q_{b'} := u_{a'}(q_{b'})$ for each child $b'$ of $a'$, then reset $u_{a'} := 0$.
    This ``lazily propagates'' the updates from a parent to a child.

    To perform an update on an interval $I$, we perform an update on each node $a$ in $\text{base}(I)$ separately.
    For each node, we first propagate as above, then modify $u_a$ and $q_a$, as above.
    We then consider the ancestors of $a$ from $a$'s parent back to the root, and update their $q$-values to maintain the invariant.
    To perform query, we perform a similar propagation on each node in $\text{base}(I)$.
    The result can then be found by combining the $q$-values of these nodes.

    We note that the segment tree need not be constructed explicitly: we only create and store values for nodes which are required.
    Using binary exponentiation, we can find the $q$-value of an untouched node in $O(\log s) = O(\log n)$ time.
\end{proof}

\subparagraph*{Hardness conjectures.}
We base hardness on the following popular conjectures.
The first is a conjecture of Henzinger et al. \cite{Henzinger2015}. 

\begin{conjecture}[OMv Conjecture]
    \label{conj:OMv}
    No (randomized) algorithm can process a given $m\times m$ boolean matrix $M$, and then in an online way compute the $(\lor, \land)$-product $Mv_i$ for any $m$ boolean vectors $v_1, \ldots, v_m$ in total time $O(m^{3-\epsilon})$, for any $\epsilon > 0$.
\end{conjecture}

In the OuMv problem, the same matrix $M$ is given during preprocessing, and $m$ pairs of boolean query vectors $(u_1, v_1), \ldots, (u_m, v_m)$ are given online.
For each, the value of the product $u_i^T M v_i$ is requested.
Note that the answer to each of these queries is a single bit.
Henzinger et al. showed that if OuMv can be solved in $O(m^{3-\epsilon})$ time, then OMv can be solved in $O(m^{3-\epsilon/2})$ time, so these problems are subcubic equivalent.

We refer the reader to the survey by Vassilevska Williams \cite{Williams2015a} for more details on the remaining conjectures.

\begin{conjecture}[APSP Conjecture]
    \label{conj:apsp}
    No (randomized) algorithm can solve {\sc All-Pairs Shortest Paths} (APSP) in $O(|V|^{3-\epsilon})$ time for $\epsilon > 0$, on graphs with vertex set $|V|$, edge weights in $\{-v^c, \dots, v^c\}$ and no negative cycles, for large enough $c$.
\end{conjecture}

\begin{definition}[$k$-OV problem]
    Let $k \geq 2$ be a constant, and $z = \omega(\log n)$.
    Given $k$ sets $A_1, \ldots, A_k \subseteq \{0, 1\}^z$ with each $|A_i| = m$, determine if there exist $a_1 \in A_1, \ldots, a_k \in A_k$ such that $a_1 \cdot \ldots \cdot a_k = 0$, where $a_1 \cdot \ldots \cdot a_k := \sum_{i=1}^z \prod_{j=1}^k a_{ji}$.
\end{definition}

\begin{conjecture}[$k$-OV Conjecture]
    \label{conj:ovc}
    \sloppy
    No (randomized) algorithm can solve $k$-OV in $m^{k-\epsilon} \textnormal{poly}(z)$ time, for any $\epsilon > 0$.
\end{conjecture}

\begin{conjecture}[3SUM Conjecture] 
    \label{conj:3sum}
    Any algorithm requires $m^{2-o(1)}$ time in expectation to determine whether a set $S \subset \{-m^3, \ldots, m^3\}$  of $m$ integers contains three distinct elements $a, b, c \in S$ with $a + b = c$.
\end{conjecture}

\section{Conditional lower bounds}
\label{sec:lowerbounds}

In this section, we establish conditional hardness for problems in $\probs'$ under popular conjectures.
We do so by considering per-operation time complexity in terms of $s$ (the side length of the grid), and overall complexity in terms of $n$ (the number of operations).

Backurs et al. \cite{Backurs2016} gave a reduction from {\sc Max $k$-Clique} to {\sc $k$D Weighted Depth}.
When $k = 3$, {\sc Max $k$-Clique} is equivalent to {\sc Negative Triangle}, which is subcubic equivalent to {\sc APSP} \cite{Williams2010}.
Adapting this reduction with a sweepline implies conditional lower bounds for \problemname{$+$}{$\max$}.

\begin{proposition}[Slightly generalised from \cite{Backurs2016}]
    \label{apsp-lb-plus-max}
    If {\sc Offline} \problemname{$+$}{$\max$} can be solved in amortised $O(s^{1-\epsilon})$ time per update and $O(s^{2-\epsilon})$ time per query, or in $O(n^{3/2-\epsilon})$ time overall, for any $\epsilon > 0$, then the \nameref{conj:apsp} is false.
\end{proposition}

\begin{proof}
    We reduce from {\sc Negative Triangle} to the equivalent \problemname{$+$}{$\min$} problem by modifying the reductions of Chan \cite{Chan2008} and Backurs et al. \cite{Backurs2016}.
    Let $G = (V, E)$ be a graph without self-loops, and with weighted edges: for each $ij \not\in E$, we let $w_{ij} = \infty$.

    We start with an $|V| \times |V|$ grid $C$ containing entirely zeros.
    Using $O(|V|^2)$ additions, we can initialise this grid to match the adjacency matrix of $G$.
    That is, the point $(i, j)$ has value $w_{ij}$.
    Then, for each vertex $i$, we add $\infty$ to point $(i, i)$.

    Next, we will iterate through each vertex $j$, in order, to test if $j$ is contained in a negative triangle.
    We iterate through each edge $ij$ in non-increasing order of $w_{ij}$.
    For each, we use two updates to add $w_{ij}$ to all points in row $i$ and column $i$ of $H$.
    At this stage, the point $(i, k)$ is equal to $w_{ik} + w_{ij} + w_{jk}$ if $i$, $j$ and $k$ are distinct vertices, and $\infty$ otherwise.
    Hence, in a single $\min$ query, we are able to determine if there is a negative triangle containing $j$.
    We then undo these row and column updates by subtracting these amounts, in the same order we added them.
    We then proceed to vertex $j+1$.

    In all, we perform $O(|V|^2)$ updates, and $O(|V|)$ queries, completing the proof.
\end{proof}

\noindent
We now establish more general linear per-operation lower bounds for {\sc 2D Grid Range} problems in terms of $s$, based on the OMv Conjecture. %, reducing from OuMv to obtain more general lower bounds.

\begin{lemma}
    \label{omv-gridwise} %\ar{broken}
    Suppose $(\mathbb{Z}, +, 0)$ is a monoid\footnote{$(\mathbb{Z}, +, 0)$ is a \emph{monoid} if $(\mathbb{Z}, +)$ is a semigroup, and the identity of $+$ is 0.} (resp. group\footnote{$(\mathbb{Z}, +, 0)$ is a \emph{group} if it is a monoid and $+$ is invertible.}), $(\mathbb{Z}, \cdot)$ is a commutative semigroup such that $0r = r0 = 0$ for all $r \in \mathbb{Z}$ and that there exists $x \in \mathbb{Z}$ such that $0 \in \{ xz, (x + x)z, (x + x + x)z \}$ if and only if $z = 0$.
    Then, \problemname{$+$}{$\cdot$} cannot be solved in worst-case (resp. amortised) $O(s^{1-\epsilon})$ time per update and $O(s^{2-\epsilon})$ time per query, for any $\epsilon > 0$, unless the \nameref{conj:OMv} is false.
    If $(\mathbb{Z}, +, 0)$ is a group, \problemname{$+$}{$\cdot$} also cannot be solved in $O(n^{3/2-\epsilon})$ time overall, for any $\epsilon > 0$, unless the \nameref{conj:OMv} is false.
\end{lemma}

\begin{proof}
    We reduce OuMv to an instance of \problemname{$+$}{$\cdot$} with $s = m$.
    Let $A_{(i, j)}$ denote the value of the point $(i, j)$.
    Initially, each $A_{(i, j)} = 0$.
    In preprocessing, for each $M_{ij} = 0$, add $x$ to $A_{(i, j)}$.
    We now say the data structure is in its \emph{ready state}.

    Let $(u, v)$ denote a pair of input vectors.
    For each $u_i = 0$, add $x$ to the value of all points in row $i$, and for each $v_j = 0$, add $x$ to the value of all points in column $j$.
    Every point now has a value in $\{0, x, x + x, x + x + x\}$.
    Now some point has value $0$ if and only if the answer to the OuMv query is 1, so we establish this with a single range query.
    We then restore the data structure to its ready state, either by keeping a journal of updates (semigroup) or by updating with additive inverses (group).
    The reduction uses $O(m^2)$ updates and $O(m)$ queries, implying the stated conditional lower bounds.
\end{proof}

This gives a $\Omega(n^{3/2-o(1)})$ time conditional lower bound on \problemname{$+$}{$\max$}, matching that of \autoref{apsp-lb-plus-max}.
Through different reductions, we are able to obtain matching lower bounds for the other problems in $\probs'$, under the same conjecture.

\begin{lemma}
    \label{omv-max}
    If any of \problemname{$\max$}{$+$}, \problemname{$\min$}{$+$} or \problemname{$\min$}{$\max$} can be solved in amortised $O(s^{1-\epsilon})$ time per update and $O(s^{2-\epsilon})$ time per query, or in $O(n^{3/2-\epsilon})$ time overall, for some $\epsilon > 0$, then the \nameref{conj:OMv} is false. 
\end{lemma}
%The proof is similar to \autoref{omv-gridwise}, with a different ready state, and is in the full version.
% \begin{proof}[Proof sketch]
%     We reduce from OuMv.
%     The proof is similar to that of \autoref{omv-gridwise}, with the key difference being the ready state.
%     When $\min$ updates are available, the ready state at the beginning of the $k$-th query is a grid where $A_{(i, j)} = -m$ if $M_{ij} = 0$, and is $-k+1$ otherwise. 
%     In each query, we set row $i$ (column $j$) to $-k$ for each $u_i = 0$ ($v_j = 0$), and check for the presence of $-k+1$.
% \end{proof}
\begin{proof}
    We will prove this for \problemname{$\min$}{$+$} and \problemname{$\min$}{$\max$}.
    The result for \problemname{$\max$}{$+$} directly follows from \problemname{$\min$}{$+$}.
   
    We reduce OuMv to an instance of \problemname{$\min$}{$q$}, with $q \in \{+, \max\}$ and $s = m$. %\ar{2 q's in the same sentence with different meanings. Should probably change the q in OuMv to n or something.}
    Let $A_{(i, j)}$ denote the value of the point $(i, j)$.
    Initially, each $A_{(i, j)} = 0$.
    In preprocessing, for each $M_{ij} = 0$, we perform an update to set the value of $A_{(i, j)}$ to $-m$.
    
    Before handling the $k$-th (1-indexed) pair of query vectors $(u, v)$, we maintain that $A_{(i, j)} = -m$ if $M_{ij} = 0$, and $-k+1$ otherwise.
    For each $u_i = 0$, we perform a $\min$ update with value $-k$ to all points in the $i$-th row.
    We perform similar updates to the $j$-th column if $v_j = 0$.
    The answer to the query is $1$ if and only if there is a point in $A$ with value $-k+1$.
    This can be established with a single $\max$ query, or with a $+$ query if we precompute the number of points where $M_{ij} = 0$.
    We finish the query by performing a single $\min$ update with value $-k$ to all points in $A$. 
    
    Over $m$ pairs of query vectors, our reduction performs $O(m^2)$ updates and $O(m)$ queries, implying the stated conditional lower bounds.
\end{proof}

\begin{lemma}
    \label{omv-set}
    If \problemname{$\setf$}{$+$} or \problemname{$\setf$}{$\max$} can be solved in $O(n^{3/2-\epsilon})$ time overall, for some $\epsilon > 0$, then the \nameref{conj:OMv} is false. %\ar{I believe bounds in terms of $s$ are 1/2 because of the $n \times n^2$ grid}
\end{lemma}
% \begin{proof}[Proof sketch]
%     We reduce from OuMv.
%     In our data structure, we have $s = m^2$ and utilise an $m \times m^2$ area of this grid, with updates and queries restricted to this area.
%     In preprocessing, we perform updates so that all points in the $j$-th column have a value of $(j-1) \mod m$.
%     Each $M_{ij}$ is represented in $A$ by a $1 \times m$ section of points with values $[0, 1, 2, ..., m-1]$. 
%     For each $M_{ij} = 0$, we perform an additional update to set its section of points to 0. 
%      
%     For the $k$-th query, we only consider columns that were assigned a value of $m-k$ during preprocessing. 
%     Among these, we set to $0$ columns $j$ where $v_j = 0$, and query rows $i$ where $u_i = 1$ to check for the presence of $m-k$.
% \end{proof}
\begin{proof}
    We reduce OuMv to an instance of \problemname{$\setf$}{$q$}, with $q \in \{+, \max\}$ and $s = m^2$. 
    We only utilise a $m \times m^2$ area of the grid, and we restrict our updates and queries to this area.
    Let $A_{(i, j)}$ denote the value of the point $(i, j)$.
    In preprocessing, we perform $m^2$ updates so that all points in the $j$-th column have a value of $(j-1 \mod m) + 1$. %\ar{assuming 1 indexing} 
    Each $M_{ij}$ is represented in $A$ by a $1 \times m$ section of points with values $[1, 2, 3, ..., m]$. 
    For each $M_{ij} = 0$, we perform an additional update to set its section of points to 0. 
    
    For the $k$-th (1-indexed) query, consider the columns that were assigned a value of $m-k+1$ during preprocessing. 
    We refer to the $j$-th such column as the \emph{$j$-th relevant column}. %\jl{Consider more precise notation for relevant columns}
    We maintain that at the beginning of each query, all points in relevant columns from previous queries have been set to 0. 
    Hence, at the start of the $k$-th query, $A_{(i, j)} \leq m-k+1$ holds for all points $(i, j)$, and $A_{(i, j)} = m-k+1$ holds only for points in relevant columns.
    Let $(u, v)$ denote a pair of input vectors.
    For each $v_j = 0$, we perform an update to set all values in the $j$-th relevant column to 0.
    The answer to the query is $1$ if and only if there is a point with value $m-k+1$ in any row $i$ where $u_i = 1$. 
    This can be established by performing a $\max$ or $+$ query in each such row. 
    After performing these queries, we set the value of the points in all relevant columns to 0.
    
    Over $m$ pairs of query vectors, our reduction performs $O(m^2)$ updates and queries, implying the stated conditional lower bound.
\end{proof}

\noindent
Together, these give conditional lower bounds for each of the single-update variants in $\probs'$.

\begin{corollary} % \ar{broken}
    If $q \in \{+, \max\}$ and $u \in \{+, \setf, \min, \max\}$ and $q \neq u$, then \problemname{$u$}{$q$} cannot be solved in worst-case $O(s^{1-\epsilon})$ time per update and $O(s^{2-\epsilon})$ time per query, or in $O(n^{3/2-\epsilon})$ time overall, for some $\epsilon > 0$, unless the \nameref{conj:OMv} is false.
    If $u \neq \setf$, then the lower bounds are amortised rather than worst-case.
\end{corollary}

\noindent

When measuring complexity in terms of $s$, it appears difficult to improve upon the naive solution which maintains a 1D instance for each column of the grid, for these problems.
However, when we measure complexity in terms of $n$, there is a polynomial gap between the $\Omega(n^{3/2-o(1)})$ time lower bound, and the $\TO(n^2)$ time naive algorithm.
Indeed, Chan \cite{Chan2008} gave a $\TO(n^{3/2})$ time solution for {\sc Grid \range{$+$}{$\max$}}.
This might lead one to ask if there exists a general mechanism to adapt $\TO(n)$ time algorithms in one dimension to $\TO(n^{3/2})$ time algorithms in two dimensions, as there is for {\sc Range} problems on a set of $n$ explicitly provided points.
Alas, when we consider variants with two simultaneous types of updates, we can obtain stronger reductions from the $d$-OV and 3SUM Conjectures, suggesting that it is unlikely that such a mechanism exists.

\begin{lemma}
    \label{ovc-lb-plus-max-sum}
    Let $d \geq 1$ be a constant.
    If {\sc Offline Static Trellised $d$D Grid \range{$\{+, \max\}$}{$+$}} or {\sc Offline Static Trellised $d$D Grid \range{$\{+, \min\}$}{$\max$}} can be solved in amortised $O(s^{(d-1)-\epsilon})$ time per update and amortised $O(s^{d-\epsilon})$ time per query, or in $O(n^{d-\epsilon})$ time overall, for any $\epsilon > 0$, then the \hyperref[conj:ovc]{$d$-OV Conjecture} is false.
\end{lemma}
\begin{proof}
    We first prove this for {\sc Grid \range{$\{+, \max\}$}{$+$}}.
    We reduce from $d$-OV with $z = \log^2 n$ to an instance of {\sc $d$D Grid \range{$\{+, \max\}$}{$+$}} with $s = m$ and $n = \TO(m)$, over the points $[m]^d$.
    Consider the first entry in each vector.
    Let $v_{ji}$ be the $i$th vector in $A_j$.
    For each $j \in [d]$ and each vector $v_{ji} \in A_j$, using the data structure, add $(v_{ji})_1$ to all points with $x_j = i$.
    Now a point $x = (x_1, \ldots, x_d) \in [m]^d$ has value $d$ if and only if $(v_{jx_j})_1 = 1$ for every $j \in [d]$.
    We then undo these updates by repeating the operations with the negations of the added values, to restore every point to 0.
    We repeat this procedure for each of the $z$ entries in the vectors.

    Now observe that $v_{1x_1} \cdot \ldots \cdot v_{dx_d} = 0$ if and only if $C_x$ never attained a value of $d$ throughout this process.
    Using a trick of Ji \cite{Ji2016} for storing the ``historic maximum'', we can use the data structure to represent an array $D$, such that at all times, $D_x$ stores the difference between $C_x$ and the maximum value attained by $C_x$ thus far.
    We can do this by translating each $+$ update to a range in $C$ to a $+$ and $\max$ update in $D$: whenever $c$ is added to $C_x$, we update $D_x \coloneqq \max(D_x - c, 0)$.
    After all updates to $C$ have been performed, $C$ is entirely zeros, so $D_x$ stores the maximum value attained by $C_x$.
    Thus, we need to check if $D$ contains any value less than $d$.
    Since every value is at most $d$, it suffices to check whether the sum of all points' values is $dm^d$.
    We note that this reduction can be done offline, uses $O(mdz) = \TO(m)$ updates (each of the form $x_j = i$) and a single query spanning the whole grid, giving the required lower bounds.
    
    {\sc Grid \range{$\{+, \min\}$}{$\max$}} is equivalent to {\sc Grid \range{$\{+, \max\}$}{$\min$}} by negating all inputs and outputs, so $D$ can be constructed similarly.
    It then suffices to check if the minimum value in $D$ is less than $d$, which can be done in a single query.
\end{proof}

\begin{lemma}
    \label{3sum-lb-plus-max-sum}
    If {\sc Offline Trellised} \problemname{$\{+, \max\}$}{$+$} can be solved with amortised $O(s^{1-\epsilon})$ time per update and query, or in $O(n^{2-\epsilon})$ time overall, and any $\epsilon > 0$, then the \nameref{conj:3sum} is false.
\end{lemma}
 \begin{proof}
     We reduce from an instance of 3SUM on a sequence $x_1, \ldots, x_m$ of $m$ integers to an instance of \problemname{$\{+, \max\}$}{$+$} with $s = m$.   %Without loss of generality, assume $s_1 \geq \ldots \geq s_n$, by sorting.
 
     Construct an $m \times m$ array $A$, such that $A_{ij} = x_i + x_j$, by using the data structure to add $x_i$ to each row $i$, and add $x_j$ to each column $j$.
     It suffices to check if $x_i$ is present in $A$, excluding row $i$ and column $i$, for any $i \in [m]$.
     This can be done for the whole of $A$ by counting the number of points in $A$ with value at least $x_i$, and value at least $x_i+1$.
     We can count and exclude the number of points in row or column $i$ with value $x_i$, which is double the number of values in the array equal to $0$, subtracting 1 for point $(i, i)$ if $x_i = 0$.
 
     Suppose we would like to count the number of points $A_{\geq c}$ in $A$ with value at least $c$.
     We can do so using a constant number of updates and queries to the data structure.
     First, request an update to replace each value $A_{ij}$ with $\max(A_{ij}, c-1)$, then calculate the new sum $\Sigma_{\geq c-1}$ of $A$.
     We can repeat this to find $\Sigma_{\geq c}$ of $A$.
     All points with original value at most $c-1$ are affected by the second update, and will increase in value by 1 between the first and the second update.
     The remaining points are unaffected by either the first or the second update.
     Thus, $A_{\geq c} = m^2 - (\Sigma_{\geq c} - \Sigma_{\geq c-1})$.
 
     This process reduces the problem to a series of $\max$ updates to $A$, each immediately followed by a sum query over $A$.
     To ensure these updates do not interfere with each other, we perform these update-query pairs in non-decreasing order of update value.
     Noting that this reduction uses $O(m)$ updates and queries completes the proof.
 \end{proof}

A trivial extension of \autoref{apsp-lb-plus-max}, together with the results above imply strong conditional hardness for {\sc Offline} \problemname{$\{+, \max\}$}{$+$}, when complexity is measured per-operation.

\begin{corollary}
    If {\sc Offline \problemname{$\{+, \max\}$}{$+$}} can be solved in amortised $O(s^{1-\epsilon})$ time per update and query, or in $O(n^{3/2-\epsilon})$ overall, for any $\epsilon > 0$, then the \hyperref[conj:apsp]{APSP}, \hyperref[conj:ovc]{2-OV} and \hyperref[conj:3sum]{3SUM} Conjectures are all false. 
\end{corollary}

Our lower bounds show that a general approach adapting almost-linear one dimensional algorithms for {\sc Grid Range} problems to truly subquadratic solutions for two dimensional instances is unlikely to exist.
However, these results do not preclude the existence of efficient and practical algorithms for specific problems, so we would like to classify which {\sc 2D Grid Range} problems can (and can't) be solved in truly subquadratic time.
To this end, we now describe truly subquadratic algorithms to some of the as-of-yet unclassified problems in $\probs'$, and generalise these to specific subclasses of {\sc 2D Grid Range} problems.

\newcommand{\staticsplit}{\texttt{partitionSlab}}

\section{Solving \texorpdfstring{{\sc Grid Range}}{Grid Range} problems with a Dynamic Partition}
\label{sec:multid-oy}

In this section, we describe an extension of the partition of Overmars and Yap \cite{Overmars1991}, and provide some applications.
Notably, our data structure is fully-online in that it does not require the coordinates in the input to be known during preprocessing.

\subsection{Dynamic Structure}

First, we introduce some terminology to reason about orthogonal objects in $d$ dimensions.

A \emph{box} is a $d$-dimensional (orthogonal) range. 
If $B = [l_1, r_1] \times \ldots \times [l_d, r_d]$ is a box and $i \in [d]$ is a dimension, the $i$-boundaries of $B$ are
\begin{align*}
 &[l_1, r_1] \times \ldots \times [l_{i-1}, r_{i-1}] \times \{ l_i \} \times [l_{i+1}, r_{i+1}] \times \ldots \times [l_d, r_d] \text{ and}\\
 &[l_1, r_1] \times \ldots \times [l_{i-1}, r_{i-1}] \times \{ r_i \} \times [l_{i+1}, r_{i+1}] \times \ldots \times [l_d, r_d].
\end{align*}
The \emph{coordinate} of an $i$-boundary is its only coordinate in the $i$th dimension.
Hence, the coordinates of the 1-boundaries of $B$ are $l_i$ and $r_i$.
If $V$ is a set of boxes, the $i$-boundaries of $V$ are the $i$-boundaries of the boxes in $V$.

For two boxes $R_1 = \prod_{i=1}^d [l_i^1, r_i^1]$ and $R_2 = \prod_{i=1}^d [l_i^2, r_i^2]$, we say that $R_1$ is an $i$-\emph{pile} with respect to $R_2$ if $[l_j^2, r_j^2] \subseteq [l_j^1, r_j^1]$ for all dimensions $j \in [d] \setminus \{i\}$.
Separately, we say that $R_1$ \emph{partially covers} $R_2$ if $\emptyset \subsetneq R_1 \cap R_2 \subsetneq R_2$.
Similarly, $R_1$ \emph{completely covers} $R_2$ if $R_2 \subseteq R_1$.

An $i$-\emph{slab} $(i \in [d])$ is a box of the form $[l_1, r_1] \times \ldots \times [l_i, r_i] \times \mathbb Z^{d-i}$.
We define $\mathbb Z^d$ to be a 0-slab. %, and we call a 1-slab a \emph{column}. % column definition not required in camera-paper. Make sure to add back if its used
A \emph{partition tree} is a rooted tree where
\begin{enumerate}
    \item All nodes are orthogonal ranges of $\mathbb Z^d$
    \item The root is $\mathbb Z^d$
    \item Every non-leaf node is the disjoint union of its immediate children.
\end{enumerate}
A partition tree is a \emph{level partition tree} if it has depth $d$, and the nodes at depth $i$ ($i \in [0, d]$) are $i$-slabs.
Hence, a node at depth $i$ in a level partition tree is \emph{cut} at a series of coordinates in dimension $i+1$ to form its children. % \jl{We don't use this term, should we CUT it? Or is it still useful for intuition?}

We now describe a data structure that maintains a list of boxes, corresponding to updates for a {\sc Grid Range} problem.
It will maintain a level partition tree of the grid so that the only boxes intersecting each leaf node are piles with respect to that leaf.
We will construct this in such a way that the number of children of each node, and number of boxes intersecting each leaf is not too large.
% We can then convert our level partition tree into a binary partition tree (with each node having at most two children), by building a segment tree between each node and its children.
To this end, we describe how to form our level partition tree.
We start with a recursive function which constructs a level partition tree over a given $i$-slab.

%At each leaf, we will store and maintain a list of boxes in $V$ which partially cover it. 
%\jl{Revisit this sentence} This allows trellised updates and queries to be performed on the structure. 
%While we do not describe a binary tree explicitly, wherever a node has more than 2 children additional intermediate nodes should be added between the node and its children, so that balanced binary tree is constructed with the old children as leaves.
%
%
%Next, we define a function which constructs a partition tree over a given $i$-slab.

\begin{algorithm}[H]
%\SetAlgoLined
\DontPrintSemicolon

\SetKwProg{Fn}{function }{}{end}
\SetKwFunction{partitionslab}{partitionSlab}

\Fn{\partitionslab{$i$, $I_1 \times \ldots \times I_i$}}{
    $s \leftarrow I_1 \times \ldots \times I_i \times \mathbb Z^{d-i}$ \;
    \If{$i = d$}{
        \Return $s$
    }
    Let $V'$ be the set of boxes in $V$ which have a $j$-boundary intersecting $s$ for some $j \leq i$ \;
    Let $\{x_1, \ldots, x_b\}$ be the set of coordinates of $(i+1)$-boundaries in $V'$, sorted so that $x_1 < \ldots < x_b$ \;
    Let $x_0 = -\infty$ and $x_{b+1} = \infty$. \;
    \For{$j = 0$ \KwTo $b$}{
        $c \leftarrow$ \partitionslab{$i+1$, $I_1 \times \ldots \times I_i \times [x_j, x_{j+1}]$} \;
        Append $c$ as a child of $s$
    }
    \Return $s$
}

\caption{Partitioning an $i$-slab}
\end{algorithm}

Suppose that $v$ is a box in $V$ and that $s$ is a node returned by \staticsplit{} at level $i$.
Then, there is at most one dimension $j \leq i$ where $v$ has $j$-boundaries intersecting $s$.
In particular, if $s$ is a leaf, then if $v$ intersects $s$, $v$ is a pile with respect to $s$.

\subparagraph*{Inserting a box.}
\label{box-insertion}
Suppose $T$ is a level partition tree over a set $V$ of boxes.
We will describe how to add a box $b$ to $T$, while maintaining this condition.

First, find the (at most) 2 columns which contain the 1-boundaries of $b$.
If either of these columns $C$ contains $2\sqrt{n}$ 1-boundaries of $V \cup \{b\}$, we delete it, and partition it into two disjoint columns $C_1$ and $C_2$, each containing at most $\sqrt{n}$ 1-boundaries of $V \cup \{b\}$.
We construct subtrees by calling \staticsplit{} on $C_1$ and $C_2$, which replace the subtrees of $C$ as children of the root.
Clearly, only $O(\sqrt{n})$ columns are created over the course of the algorithm.

%We use the above construction, with rebuilding, to create a dynamic partition of $\mathbb Z^d$ into t-regions.
%Let $V$ be the set of boxes that have been inserted in the structure at some point in time. 
%The partition is stored in a single tree, with the root (depth 0) representing all of $\mathbb Z^d$.
%The root has at most $O(\sqrt{n})$ children, each slabs at level 1 (called `columns'). 
%We maintain that these columns each contain at most $2\sqrt{n}$ 1-boundaries of $V$.
%The subtrees of these columns match the construction described above exactly. 

Next, for each $i \in [1, d-1]$, let $S_i$ be the set of slabs $s_i$ in $T$, such that a $j$-boundary of $b$ intersects $s_i$ for some $j \leq i$.
To maintain the condition, we need to cut each $s_i$ at the coordinates of the $(i+1)$-boundaries of $b$. 
This splits (up to) two children of $s_i$ into two new slabs at depth $i+1$: we once again call \staticsplit{} on each of these new slabs and replace the subtree rooted at $s_i$ with the returned trees.

\subparagraph*{Balancing boxes.} 
The final ingredient in our algorithm is a mechanism for bounding the maximum number of children of a node.
To do so, we strategically insert additional boxes, called \emph{balancing boxes}. %, to ensure the required conditions for \autoref{static-construction-complexity} and \autoref{dynamic-OY-complexity} are met.
Let $C$ be a column of $T$, $i$ be a dimension and $x$ be a coordinate in dimension $i$.
A $(C, i, x)$-balancing box has the following properties:
\begin{itemize}
    \item The box has both of its 1-boundaries inside the column $C$. 
    \item The $i$-boundaries of the box both have coordinate $x$;
    \item The remaining $j$-boundaries ($j \not\in \{1, i\}$) are chosen arbitrarily.
\end{itemize}
For each dimension $i$, we keep a list $X_i$ of the coordinates of $i$-boundaries of $V$ in non-decreasing order.
For each column $C$, we maintain the invariant that for each consecutive range $X_i[j, j+\sqrt{n}-1]$ of $\sqrt{n}$ $i$-boundaries in $X_i$, there is at least one $(C, i, x)$-balancing box with $x \in [X_i[j], X_i[j+\sqrt{n}-1]]$.
If after the insertion of a box there exists a range where this condition is not met, we add a balancing box with $x$ being the $i$-coordinate in the centre of the range. 
Hence, $O(d\sqrt{n})$ balancing boxes are added for each column, and so $O(n)$ balancing boxes are added overall, since $d$ is a constant.

This concludes the description of the algorithm; we now analyse its time complexity.
First, we make an observation about the impact of balancing boxes.

%When a balancing box is added, a separate $(c, i, x)$-balancing box is added for every column $c$.
%The insertion of a balancing box is done in the same way as any other box. 
%When a new column is created, all pre-existing balancing boxes are inserted with this new column as $c$.
%Note that $O(n)$ balancing boxes are created across all insertions.
\begin{lemma}
    \label{balancing-boxes}
    When balancing boxes are used, the set $V'$ described in \textnormal{\staticsplit{}} will be of size $O(\sqrt{n})$ for all nodes in the partition tree.
\end{lemma}
\begin{proof}
    Let $B_C$ be the set of balancing boxes with a column $C$.
    By definition, all $b \in B_C$ have a 1-boundary in $C$.
    Consider all slabs $s_i$ in the subtree of $C$ at a depth $i < d$. 
    Since a cut is made at the $(i+1)$-boundaries of $b \in B_C$, the children of $s_i$ each contain at most $O(\sqrt{n})$ $(i+1)$-boundaries, and so by induction, this is true for all nodes in the subtree rooted at $s_i$.
\end{proof}

We now obtain the following generalisation of the result of Overmars and Yap.

\begin{theorem}
    \label{t-regions}
    There is a data structure that maintains a set $V$ of boxes, and supports $n$ fully-online box insertions to $V$.
    Throughout, it can maintain a level partition tree of $\mathbb Z^d$ whose leaves partition $\mathbb Z^d$ into a set of $O(n^{d/2})$ axis-aligned regions (colloquially ``t-regions'', short for ``trellised-regions'') such that:
    \begin{enumerate}
        \item Every box in $V$ does not intersect, completely covers or is a pile with respect to each t-region;
        \item Each box partially covers $O(n^{(d-1)/2})$ t-regions; 
        \item Each t-region is partially covered by at most $O(\sqrt{n})$ boxes;
        \item Any line parallel to a coordinate axis intersects at most $O(n^{(d-1)/2})$ t-regions; and
        \item A list of the boxes that partially cover each t-region is maintained. 
    \end{enumerate}
    This all can be done in amortised $\TO(n^{(d-1)/2})$ time per insertion.
\end{theorem}

% omitted - from camera ready paper
%\begin{proof}[Proof sketch]
%    Overmars and Yap \cite{Overmars1991} construct such a partition tree during precomputation from the boxes' coordinates, however this is not possible in a fully-online setting.
%    Instead, we maintain the required properties as boxes are added by periodically rebuilding subtrees which exceed a certain size and strategically inserting additional boxes to preserve balance.
%    The cost of rebuilding amortises over the $n$ operations, giving the stated time complexities.
%\end{proof}

\begin{proof}
    It follows from \autoref{balancing-boxes} that in each execution of \staticsplit{}, $|V'|$ $= O(d\sqrt{n})$ $= O(\sqrt{n})$.
    Hence, each recursive call in \staticsplit{} creates $O(\sqrt{n})$ children.
    
    We will now prove the five conditions of \autoref{t-regions} are met. 
    The first condition is true due to the method used to partition slabs. 
    In particular, performing cuts for all boxes in $V$ which have a $j$-boundary intersecting $s$ for some $j \leq i$ (that is, the set $V'$) ensures that only one $j$-boundary of a given box can intersect each t-region.
    
    For the second condition, consider some box $b$.
    For $b$ to partially cover a t-region, it must have an $i$-boundary in the region for some dimension $i$.
    Consider all nodes $s$ in the tree with a depth $i-1$. 
    There are $O(n^{(i-1)/2})$ such nodes $s$, since the degree of all nodes in the tree is bounded by $O(\sqrt{n})$.
    At most 2 children of $s$ contain an $i$-boundary of $b$. 
    There are $O(n^{(d-i)/2})$ leaves in the subtrees of these children.
    Hence, at most $O(n^{(i-1)/2} \times n^{(d-i)/2}) = O(n^{(d-1)/2})$ t-regions are partially covered by $b$ for each dimension $i$.  
    
    The third condition follows from \autoref{balancing-boxes}.
    
    The fourth condition can be proved in a similar way to the second.
    Consider a line parallel to the $i$-axis. 
    For each node at a depth $i-1$ in the partition tree, the line can only intersect one of its children.
    
    For the fifth condition, each time a new t-region is created during \staticsplit{}, we create a list of the boxes which partially cover it.
    Additionally, each time a new box $b$ is inserted we append it to the lists of all pre-existing t-regions that it partially covers. 
    Note that by the second condition, $b$ is added to $O(n^{(d-1)/2})$ lists. 
    This can be done in $O(n^{(d-1)/2})$ time by iterating over all nodes at a depth of $d-1$ in the tree, and finding which of their children are partially covered by $b$.
    
    We will now consider the time complexity of \staticsplit{}.
    Since the returned tree has depth $d-i$, there are at most $O(n^{(d-i)/2})$ leaves in the partition.
    In each call to \staticsplit{}, we can find the set $V'$ in $\TO(|V'|) = \TO(\sqrt{n})$ time by keeping a sorted list of $i$-boundaries for each dimension $i$.
    We can create the list of boxes which cover a t-region in the same way.   
    Hence, \staticsplit{} on an $i$-slab has a time complexity of $O(n^{(d-i+1)/2})$.
    
    Now consider the complexity of producing the partition. 
    Across all $n$ insertions, $O(\sqrt{n})$ columns are created, each in $\TO(n^{d/2})$ time.
    This takes $\TO(n^{(d+1)/2})$ time overall and hence amortised $\TO(n^{(d-1)/2})$ time per insertion.
    
    When a box is inserted, up to two columns are split.
    Additionally, let $S$ be the union of the sets of slabs $S_i$ described in the \hyperref[box-insertion]{\emph{inserting a box}} paragraph.
    We will now consider the contribution of splits made to slabs in $S$.
    If a slab is not in $S$, at most 2 of its children are in $S$.
    Hence, for each dimension $i$, there are at most $O(n^{(i-1)/2})$ slabs that are split.
    The overall cost of calling \staticsplit{} on these slabs is $O(n^{(d-1)/2})$.
\end{proof}

Using this structure, we can reduce {\sc Grid Range} problems to {\sc Trellised Grid Range} problems when the update operation distributes over the query operation.

\begin{theorem}
    \label{reduce-to-trellised}
    Suppose $q$ and $u$ are associative operations, both computable in $O(1)$ time, such that $q$ is commutative, $u$ distributes over $q$, and 0 is an identity of $u$.
    If {\sc Trellised $d$D Grid \range{$u$}{$q$}} can be solved in $\TO(n)$ time, then {\sc Grid \range{$u$}{$q$}} can be solved in $\TO(n^{(d+1)/2})$ time.
\end{theorem}

\begin{proof}
    We use the structure given by a dynamic level partition tree $J$ from \autoref{t-regions}.
    Recall that there are $O(n^{(d-1)/2})$ $(d-1)$-slabs in $J$, and each is a parent of $O(\sqrt{n})$ t-regions, which are leaves of $J$. % to a binary partition tree $J'$, by building and maintaining a segment tree between each node of $J$ and its immediate children.
    For each $(d-1)$-slab, we maintain a data structure supporting range updates and queries within the slab.
    Conceptually, we maintain a 1D array over its children, in order of their $d$-coordinate, with each array entry containing a {\sc Trellised} instance for the corresponding t-region.
    We support operations whose ranges completely cover multiple children by using a modified version of the {\sc 1D Grid \range{$u$}{$q$}} structure (see \autoref{lazyprop}), treating the $q$-value of entire regions as the array values.
    Operations affecting only part of a child are handled by the corresponding {\sc Trellised} instance.
    In particular, when a box corresponding to an update (query) is inserted into the list of a t-region, we perform an update (query) on that t-region.

    Operations are then performed by iterating over all $(d-1)$-slabs, and updating and querying the respective data structures, as necessary.
    Our data structure is maintained in such a way that over its lifetime, there will be $O(n^{d/2})$ {\sc Trellised} instances, each facilitating $O(\sqrt{n})$ operations, giving the required time complexity.
\end{proof}

\subsection{Applications}

We apply \autoref{reduce-to-trellised} to give $\TO(n^{(d+1)/2})$ time algorithms for particular problems.

First, we show that {\sc Trellised} variants can be solved as separate one dimensional instances, when the conditions of \autoref{reduce-to-trellised} are met and $u$ is also \emph{commutative}.

\begin{lemma}
    \label{semiring-trellised}
    Suppose $q$ and $u$ are associative, commutative binary operations, computable in $O(1)$ time, such that $u$ distributes over $q$, and 0 is an identity of $u$.
    Then {\sc Trellised Grid \range{$u$}{$q$}} is solvable in $\TO(n)$ time.
\end{lemma}

\begin{proof}
    For notational convenience we provide a proof for the case where $q$ is $\max$ and $u$ is $+$: other operations are proven identically.
    By the definition of {\sc Trellised}, we can associate every update with a dimension $i$ such that the update range is an $i$-pile with respect to $\mathbb Z^d$; we call this an \emph{$i$-update} for short.
    At any given point in time, let $U_i(x)$ be the sum (with respect to $u$) of all $i$-updates with coordinate $x$ in dimension $i$.

    Now consider a query over the range $R = [l_1, r_1]\times\ldots\times[l_d, r_d]$. 
    The answer to the query can be written as
    \[
        \max_{(x_1, \ldots, x_d) \in R} \sum_{i \in [d]} U_i(x_i)
        =
        \sum_{i \in [d]} \max_{(x_1, \ldots, x_d) \in R} U_i(x_i)
        =
        \sum_{i \in [d]} \max_{x_i \in [l_i, r_i]} U_i(x_i)
    \]
    by distributivity.
    Hence, we can reduce to $d$ instances of {\sc 1D \range{$u$}{$q$}}, which each can be solved in $\TO(n)$ time, by \autoref{lazyprop}.
\end{proof}

\noindent
This proves \autoref{thmt@@conditionsoy}, giving efficient fully-online algorithms in these cases.

\begin{corollary}
    {\sc Grid \rangeqint{$+$}{$\max$}} and {\sc Grid \rangeqint{$\min$}{$\max$}} can be solved in $\TO(n^{(d+1)/2})$ time.
\end{corollary}

\noindent
There also exists some instances where $u$ is not commutative for which an $\TO(n)$ solution to {\sc Trellised Grid \rangeqint{$u$}{$q$}} exists, giving us algorithms with the same time complexity.
\begin{lemma}
    \label{(max, set) trellised}
    {\sc Trellised Grid \rangeqint{$\setf$}{$\max$}} can be solved in $\TO(n)$ in $d$ dimensions.
    % Hence, {\sc Grid \rangeqint{$\setf$}{$\max$}} can be solved in $\TO(n^{(d+1)/2})$.
\end{lemma}
\begin{proof}
    Consider only $i$-updates for some dimension $i$. 
    For each $i$-coordinate, we will store the latest $\setf$ update which has been applied to it.
    We compress this by storing a structure $R_i$ of non-intersecting segments, where each segment corresponds to contiguous $i$-coordinates which have had the same latest update applied to them. 
    We can maintain this as new updates occur in $\TO(n)$ across all updates.
    
    Consider a query on a box $b$.
    For each dimension $i$, there will be a range of segments in $R_i$ which lie within or intersect $b$.
    Let $t_i$ be the earliest update time of all such segments.
    Let $t$ by the maximum among all $t_i$.
    
    Consider a segment in $R_i$ (for any dimension $i$) which lies within or intersects $b$.
    Its update is visible at some point in $b$ if and only if its update occurred at or after time $t$.
    The answer to the query is the maximum value among all such segments.
    This can be computed in $\TO(1)$ using 2D segment trees that support point updates and range maximum queries. %\ar{a simpler way to handle this is to assume queries are trellised, although does not fit the definition of trellised (max, set)}
\end{proof}
\begin{corollary}
    \label{(max, set)}
    {\sc Grid \rangeqint{$\setf$}{$\max$}} can be solved in $\TO(n^{(d+1)/2})$ time.
\end{corollary}

\section{Reducing to \texorpdfstring{{\sc Static Trellised}}{Static Trellised} instances}
\label{sec:minmax}
The technique from the previous section does not work on all variants.
For instance, consider the \problemname{$\{\min, \max\}$}{$\max$} problem.
While the update operations ($\min$ and $\max$) are individually associative, commutative and distribute over the query operation ($\max$), they do not commute with each other.
Given that the order of operations matters greatly, it is difficult to decompose this into separate one dimensional problems.
As a means for solving such problems, in this section we present a general framework for reducing multidimensional range problems to {\sc Static Trellised} instances.

\subsection{A sketch in two dimensions}
\label{subsec:sketch}

We will first give a high-level overview of our reduction framework, using \problemname{$\{\min, \max\}$}{$\max$} as an example.
For simplicity, we assume these instances are {\sc Offline} and in two dimensions.
We formalise this framework in \autoref{subsec:formal}, also generalising our results to {\sc Online} and higher dimensional settings, accompanied by full details and proofs.

\subparagraph*{General approach.}
Our algorithm operates by partitioning operations into chronologically contiguous batches of at most $k$ operations, for some function $k$ of $n$.
Each batch may contain both updates and queries.
Let $B$ be a particular batch, and let $G_B$ be the state of the grid at the start of $B$.
We will show how to answer the queries within $B$.

The $\bar{k} = O(k)$ coordinates of $B$'s operations partition the grid into a $\bar{k} \times \bar{k}$ \emph{overlay grid} of \emph{overlay regions}.
Any update or query will concern all points that fall within a 2D range of whole overlay regions.
Since the update operations $\min_c$ and $\max_c$ are monotonically increasing functions for any constant $c$, it suffices to know the maximum value within each overlay region according to $G_B$, to answer the queries of $B$.
We use these maximums as initial values for a \problemname{$\{\min, \max\}$}{$\max$} instance over the overlay grid, which we solve by keeping $\bar{k}$ {\sc 1D Range} instances, one for each column, and facilitating operations in $\TO(k)$ time, by iterating over each one.
% JL: Partitioned range not necessary since it's offline, so no partitioning =).

It remains to find the maximum value within each overlay region, according to $G_B$.
To do so, consider an alternate partition of the grid into t-regions according to \autoref{t-regions}.
In each t-region $Y$, we will form an instance of {\sc Static Trellised \problemname{$\{\min, \max\}$}{$\max$}} with $O(\sqrt{n})$ updates; we describe how to do so below.
Then, to find the maximum value within an overlay region $O$, we issue queries to the instances corresponding to the t-regions intersecting $O$.

\subparagraph*{Forming {\sc Static Trellised} instances.}
By construction, each t-region $Y$ is affected by up to $n$ \emph{whole updates} which completely cover $Y$, and up to $\sqrt{n}$ \emph{partial updates} which cover all points in a range of rows or range of columns of $Y$.
We can afford to include each partial update in our instance, but need to find a succinct way to represent whole updates.

To do so, we construct a segment tree $T$ over time.
Both $\min$ and $\max$ operations can be written as functions of the form $f_{a, b}(x) = \min(a, \max(b, x))$, which are closed under composition.
For each range of time $[t_1, t_2]$ in the tree, we will build a data structure which, given a t-region $Y$, returns values $a$ and $b$ such that the composition of whole updates to $Y$ occurring during $[t_1, t_2]$ is $f_{a, b}$.
Such a data structure can be built by performing lazy updates over the t-region tree.
Each update affects $O(\sqrt{n})$ nodes in the t-region tree, and each whole update is present in $O(\log n)$ ranges of $T$.
$T$ can then return an $f_{a, b}$ for arbitrary ranges of time and t-regions.

Each t-region $Y$ has $\sqrt{n} + 1$ ranges of time before and after each of its partial updates.
For each of these, we query $T$ and translate the returned $f_{a, b}$ into a $\min$ and $\max$ update spanning all of $Y$.
Hence, we have formed an instance of {\sc Static Trellised \problemname{$\{\min, \max\}$}{$\max$}} with $O(\sqrt{n})$ updates, as required.
This concludes the description of our approach.
We now turn our attention to analysing its time complexity.

\begin{lemma}
    \label{minmax-reduction}
    If there an algorithm for {\sc Static Trellised \problemname{$\{\min, \max\}$}{$\max$}} running in $\TO(\numupdates^{c-\gamma} (\numupdates+\numqueries)^\gamma)$ time for some $c \geq 1$ and $\gamma \in [0, 1]$, then \problemname{$\{\min, \max\}$}{$\max$} can be solved in $\TO(n^{5/4 + c/2})$ time.
\end{lemma}

\begin{proof}
    We process each of the batches separately, and consider the time taken to answer the queries within a particular batch $B$.

    Using the methods above, in $\TO(n^{3/2})$ time we form an instance of {\sc Static Trellised \problemname{$\{\min, \max\}$}{$\max$}} with $\numupdates = O(\sqrt{n})$ updates in each of $O(n)$ t-regions.
    Next, we bound the number of queries made to such instances within $B$.
    Consider the partition $\mathcal P$ of the grid produced by refining each t-region by the overlay grid.
    A query to a {\sc Static Trellised} instance is made for each region in $\mathcal P$, and the number of these regions and thus, queries, is $O((k + \sqrt{n})^2) = O(k^2 + n)$: this is the number of intersections found when the overlay grid is laid atop the t-regions.

    For a given t-region $Y$, let the number of updates and queries made to the instance in $Y$ be $n_{{\textnormal u}_Y}$ and $n_{{\textnormal q}_Y}$, respectively.
    First, consider the regions $y$ where $n_{{\textnormal q}_y} < \sqrt{n}$.
    There are $O(n)$ regions in total, so we spend $\TO(n^{1+c/2})$ time answering queries for these regions.

    Now consider the regions $\Upsilon$ where $n_{{\textnormal q}_\Upsilon} \geq \sqrt{n}$.
    Since $n_{{\textnormal q}_\Upsilon} = \Omega(\numupdates)$, the running time of $\Upsilon$'s {\sc Trellised} instance is $\TO(n_{{\textnormal u}_\Upsilon}^{c-\gamma} n_{{\textnormal q}_\Upsilon}^\gamma)$, which is subadditive with respect to $n_{{\textnormal q}_\Upsilon}$.
    Thus, the total running time for these regions is maximised when the queries are distributed evenly among as many regions as possible.
    Subject to the constraint on these regions, this maximum is achieved, within a constant multiplicative factor, when there are at most $O((k^2 + n)/\sqrt{n})$ regions, each with $\sqrt{n}$ queries.
    Hence, in this case, the total running time is bounded by $\TO((k^2 + n)n^{(c-1)/2})$.
    We thus spend time $\TO(n^{3/2} + n^{1+c/2} + k^2 n^{(c-1)/2})$ for each of $n/k$ batches.
    For balance, we choose $k = n^{3/4}$, giving a running time of $\TO(n^{5/4 + c/2})$ overall.
\end{proof}

\sloppy
It remains to show that we can solve {\sc Static Trellised Grid} \range{$\{\min, \max\}$}{$\max$} efficiently.
To do so, we first describe the following ancillary data structure.

\begin{lemma}
    \label{buckets}
    If $S$ is a (multi)set of integers, let $L(S)$ denote the value in $S$ with largest absolute value.

    Let $S_1, \dots, S_n$ be multisets of integers.
    Then, there exists a data structure that can facilitate any $m$ of the operations:
    \begin{itemize}
        \item $\text{add}(l, r, x)$: for each $i \in [l, r]$, add a single copy of $x$ to $S_i$
        \item $\text{remove}(l, r, x)$: for each $i \in [l, r]$, remove a single copy of $x$ from $S_i$; it is guaranteed that a previous $\text{add}(l, r, x)$ call has been made, and that a corresponding $\text{remove}$ call has not yet been made%\ar{I think requirement on remove should be stricter for worst case complexites}
        \item $\text{query}(l, r, x)$: return $\max_{i \in [l, r]} L(S_i)$
    \end{itemize}
    in $O(\log n \log m)$ time for add and remove, and $O(\log n)$ time for query, in the worst case.
\end{lemma}

\begin{proof}
    Construct a segment tree, and keep an ordered set $S_u$ (e.g. a self-balancing binary search tree) at each node $u$.
    We maintain a value $L^{\ast}_v$ at each node $v$, such that at all times, $L^{\ast}_v = L(\{L(S_v), \max(L^{\ast}(S_a), L^{\ast}(S_b))\})$, where $a$ and $b$ are the immediate children of $v$.

    Recall (from \autoref{sec:prelim}) that any range $[l, r]$ can be written as the disjoint union of a set $\text{base}([l, r])$ of $O(\log n)$ segment tree ranges.
    For the add (remove) operation, add (remove) $x$ from the sets corresponding to the nodes in $\text{base}([l, r])$.
    Each addition or removal from a set takes $O(\log m)$ time.
    Next, update the $L^{\ast}$ values of these nodes, and their ancestors.
    To answer a query, it suffices to consider the $L^{\ast}$ values of the nodes in $\text{base}([l, r])$, together with the $L$ values of their ancestors and use a recurrence identical to what is used to update the tree.
\end{proof}

\begin{lemma}
    \label{minmax-st}
    {\sc Static Trellised} \problemname{$\{\min, \max\}$}{$\max$} can be solved in $\TO(u)$ time.
\end{lemma}

\begin{proof}
    We describe how to preprocess a data structure that can answer queries of the form: ``Is there a value in a given range $[l_1, r_1] \times [l_2, r_2]$, greater than or equal to a given value $x$?'', where the range and $x$ vary per query.
    Each query can then be answered with binary search.

    Observe that for any fixed $x$, we need only to consider two types of updates: $\min$ updates with value less than $x$, and $\max$ updates with value at least $x$.
    If an update occurred at time $t$, we assign it a {\em label} of $-t$ if it is of the first kind, and $+t$ if it is of the second.
    We say that a label is {\em later} than another if its absolute value is larger.
    
    Consider a row (column) labelled by the labels of all row (column) updates covering it.
    Let $L_r$ ($L_c$) be the latest label of row $r$ (column $c$).
    The answer to the query is then ``Yes'' if and only if there is a point $(r, c) \in [l_1, r_1] \times [l_2, r_2]$ such that either $L_r > 0$ and $|L_r| > |L_c|$, or $L_c > 0$ and $|L_c| > |L_r|$.
    However, this is precisely when $L_r + L_c > 0$. % \ar{remove absolute values}
    As $r$ and $c$ can be chosen independently, it suffices to find $\max_{r \in [l_1, r_1]} L_r$ and $\max_{c \in [l_2, r_2]} L_c$, and check if their sum is greater than 0.

    We show how to find the former value: the same can then be done for the latter.
    First, create an instance of the \autoref{buckets} data structure with a multiset for each row.
    We start with $x = -\infty$, adding all $\max$ updates into the data structure, with their positive labels.
    Next, we sort all row updates (both $\min$ and $\max$) into non-decreasing order by update value, and iterate through this sorted list.
    When we encounter a $\min$ update, we add its label to its corresponding range of rows, and conversely, when we encounter a $\max$ update, we remove its label from its corresponding range of rows.
    In this way, each update corresponds to an increase in $x$. %\ar{I think it should be: when we encounter a min update, we add its label (rather than other way around)}
    By making this data structure persistent \cite{Driscoll1986}, we can facilitate online queries for \emph{any} $x$ with only constant multiplicative overhead in running time.
\end{proof}

By observing that $\setf_z = \min_z \circ \max_z$, we obtain the following result.

\begin{corollary}
    %\label{minmax-st}
    \label{minmax-2d}
    \problemname{$\{\min, \setf, \max\}$}{$\max$} can be solved in $\TO(n^{7/4})$ time.
\end{corollary}

\subsection{Online and higher dimensional settings}
\label{subsec:formal}

Most of the steps in the approach of the previous subsection are not specific to the update and query operations in our example.
To generalise this approach to fully-online and multidimensional settings, we introduce the following ``partial information'' variant of {\sc 1D Range}.

\begin{definition}[{\sc 1D Partitioned Range $(q, U)$}]
    Let $A_0$ be an integer array of length $s$ and let $0 = a_0 < \dots < a_\rho = s$ be a sequence of indices.
    Let $Q$ be a corresponding sequence of $\rho$ integers, denoting the values $q(A_0[a_0 + 1, a_1]), \ldots, q(A_0[a_{\rho-1} + 1, a_\rho])$.
    Initially, $\rho = 1$.

    Let $J$ be a list of range updates applicable to $A_0$ in chronological order, initially empty.
    We write $A_J$ for the result of applying the updates of $J$ to $A_0$, in order.

    \sloppy
    Given an integer $s$, and the value of $q(A_0[1, s])$, support the following operations:
    \begin{itemize}
        \item $\text{split}(i, a, q_l, q_r)$: given $i \in [0, \rho-1]$ and $a_i < a < a_{i+1}$, add $a$ to the sequence of indices, and update $Q$ with the knowledge that $q(A_0[a_i + 1, a]) = q_l$ and $q(A_0[a + 1, a_{i+1}]) = q_r$
        \item $\text{update}_j(a_l, a_r)$: append to $J$, the update: ``for each $i \in [a_l+1, a_r]$, set $A[i] := u_j(A[i])$''
        \item $\text{query}(a_l, a_r)$: return $q(A_J[a_l+1, a_r])$
    \end{itemize}
    It is guaranteed that every $a_l$ or $a_r$ provided as input will already be in the sequence.
\end{definition}

Note that this problem is not solvable for all choices of $q$ and $U$: one may need to know the individual values of $A_0$, and not just the result of $q$ over some ranges, to facilitate queries after certain types of updates.
To see this, consider the case when $q = +$ and $U = \{max\}$, and suppose that we know that $q(A_0[1, s]) = A_0[1] + \ldots + A_0[s] = 3$.
After we perform $\max_5$ over the entire array, we do not have enough information to facilitate a query for the sum of the new array: for instance, we would need to know the number of elements with value at most 5, and the sum of elements greater than 5.

We use the {\sc 1D Partitioned Range} problem to describe a more general result.

\begin{restatable}{theorem}{onlinemultidstatictrellis}
    \label{general-static-trellis}
    % \label{online-multid-static-trellis}
    Suppose $U$ is a set of update functions, and $q$ is a query function, computable in $O(1)$ time.
    If there is a set $\bar{U}$ such that:
    \begin{enumerate}
        \item $U \subseteq \bar{U}$ are sets of functions that can be represented and composed in $\TO(1)$ space and time, such that the composition of any series of at most $n$ (possibly non-distinct) functions of $U$ results in a function in $\bar{U}$; %\ar{conditions on $\bar{U}$?}
        \item There is an algorithm for {\sc 1D Partitioned Range $(q, U)$} that performs both updates and queries in $\TO(1)$ time;
        \item There is an algorithm for {\sc Static Trellised $d$D Grid \range{$\bar{U}$}{$q$}} that runs in $\TO(\numupdates^{c-\gamma} (\numupdates+\numqueries)^\gamma)$ time for some $\gamma \in [0, 1]$
    \end{enumerate}
    then {\sc Grid \range{$U$}{$q$}} can be solved in $\TO(n^{\frac{c+d+1}{2} - \frac{1}{2d}})$ time.
    When $d = 2$, this is $\TO(n^{5/4 + c/2})$ time.
\end{restatable}

\noindent
We dedicate the remainder of this section to proving this result, which will generalise \autoref{minmax-2d}.

To prove this, we adopt a batching approach once again, and precompute data structures at the start of every batch which allow us to facilitate the operations within the batch in and online manner.
%For ease of calculations, assume $k > \sqrt{m}$ from here on. \jl{This is awkwardly placed}
%Despite what this name might imply, we still handle operations in an online fashion: we respond to any query before the next arrives.
%This is accomplished by performing a preprocessing step at the conclusion of each batch, which builds a data structure summarising the updates that have been performed thus far.
%%We describe this preprocessing below.
%\ar{says maximum value although we are talking in general}
Throughout the batch, we partition the grid according to the \emph{overlay grid}, whose construction is dependent on the operations in the batch.
We formalise this as follows.
At any given point in time, we maintain, for each each dimension $i$, a sorted sequence $C_i$ of distinct coordinates in dimension $i$, with $C_{i, 0} = 0 < C_{i, 1} < \ldots < C_{i, |C_i|} = s$.
This partitions the coordinates of dimension $i$ into disjoint ranges $\mathcal{R}_i = \{[C_{i, j} + 1, C_{i, j+1}] : j \in [0, |C_i|-1]\}$.
In turn, these $\mathcal{R}_i$'s partition the grid into a set of disjoint \emph{overlay regions} $\mathcal{O} = \{ R_1 \times \ldots \times R_d : (R_1, \ldots, R_d) \in \mathcal{R}_1 \times \ldots \times \mathcal{R}_d \}$.
%Given a sequence of row indices $r_0 = 0 < r_1 < \dots < r_p = n$ and column indices $c_0 = 0 < c_1 < \dots < c_q = n$, we can partition the points of an $n \times n$ grid into a set of disjoint \emph{overlay rows} $\mathcal{R} = \{[r_i + 1, r_{i+1}] \mid i \in [0, p-1]\}$, \emph{overlay columns} $\mathcal{C} = \{[c_j + 1, c_{j+1}] \mid j \in [0, q-1]\}$, and \emph{overlay regions} $\mathcal{O} = \{ R \cap C \mid R \in \mathcal{R} \text{ and } C \in \mathcal{C}\}$.
We refer to this structure as an \emph{overlay grid}.
At the beginning of a batch, we have each $C_i = \{ s\}$, each $\mathcal{R}_i = \{[1, s]\}$ and so $\mathcal{O} = \{[1, s]^d\}$.

When we receive an update or query for a range $[l_1, r_1] \times \ldots \times [l_d, r_d]$, for each dimension $i$, we add $l_i - 1$ and $r_i$ to $C_i$.
There are thus at most $\bar{k} = 2k + 1$ ranges in each dimension per batch, and each added coordinate splits a range in $\mathcal{R}_i$ thereby refining the partition $\mathcal{O}$ according to a $d-1$ dimensional hyperplane.
After these additions, the update or query at hand concerns precisely a range of overlay regions within the overlay grid.

Let the \emph{initial value} of an overlay region $R$ be $q(R)$ at the start of the batch.
Our algorithm consists of two parts: first, maintaining the initial value of each overlay region as $\mathcal{O}$ changes; and using these values to support updates and queries within this batch.

\subparagraph*{Maintaining the initial value of overlay regions.}
We perform the following precomputation at the start of each batch.
Let $V$ be the set of update ranges appearing in prior batches.
Separate to the overlay grid, we use the following result of Overmars and Yap \cite{Overmars1991} to construct a separate partition $\mathcal{Y}$ of the grid into \emph{t-regions} according to $V$.
This theorem is a static version of \autoref{t-regions}, so we also call these \emph{t-regions}, since they bear the same properties.

\begin{lemma}[\cite{Overmars1991}]
    \label{overmars-yap}
    A partition of the grid into t-regions can be constructed from a collection $V$ of $n$ $d$-dimensional ranges in $O(n^{(d+1)/2})$ time, with the following properties:
    \begin{enumerate}
        \item There are $O(n^{d/2})$ t-regions;
        \item Any hyperplane orthogonal to a coordinate axis intersects at most $O(n^{(d-1)/2})$ t-regions. Hence, each range of $V$ partially covers at most $O(n^{(d-1)/2})$ points;
        \item Each t-region only contains piles (with respect to itself) in its interior; and
        \item Each t-region has at most $O(\sqrt{n})$ ranges of $V$ partially (but not completely) covering it.
    \end{enumerate}
\end{lemma}

Within each t-region $Y$, we will construct an instance of {\sc Static Trellised $d$D \range{$U$}{$q$}}, with $O(\sqrt{n})$ updates.
We do so in the same way as in \autoref{subsec:sketch}: each pile intersecting $Y$ will translate to a trellised update.
Updates completely covering $Y$ can be summarised as $O(\sqrt{n})$ updates by building a segment tree over time.
In each segment tree range, we will perform lazy propagation over the t-region tree.
This takes $\TO(n^{(d+1)/2})$ time in total.

%For any $Y' \in Y$ and any rectangular range $R'$, let $I_{Y'}(R')$ denote the answer to the query for $R'$, on the {\sc Static Trellised \problemname{$\setf$}{$+$}} instance for $Y'$.
%Since $Y$ partitions the grid, for \emph{any} range $R$, the initial value of $R$ can be written as $\sum_{Y' \in Y} I_{Y'}(R \cap Y')$.
%This gives us a simple way to find the initial value for each overlay region by querying our constructed instances.
%Unfortunately, the naive implementation of this approach is inefficient as there can be $O(m)$ ranges $Y'$ that intersect $R$.
%One optimisation is to store the initial value of each $Y'$ in a data structure, and query this data structure to account for all $Y'$ that are entirely contained in each $R$.
%This method generates at most $O(k^2 \sqrt{m})$ queries, and there are simple constructions that produce as many; this is still too many to produce a subquadratic algorithm overall.
%We take this optimisation a step further, and store the result of \emph{all} queries we have made that may be potentially useful to us in the future.

Next, define the set of \emph{fragments} to be $\mathcal{F} = \{ Y \cap O : Y \in \mathcal{Y} \text{ and } O \in \mathcal{O} \}$.
$\mathcal{F}$ can be seen as a refinement of our partitions $\mathcal{Y}$ and $\mathcal{O}$, and so is itself a partition of the grid into disjoint ranges.
At the beginning of the batch, the set of \emph{initial fragments} is $\mathcal{Y}$.
When a coordinate is added to $C$, some fragments are split, with each split creating two (smaller) new fragments.
We will maintain the initial value of each fragment throughout the batch.
%The idea is to compute and store the intial value of each fragment in a data structure that supports range queries.
%By definition, each fragment belongs to precisely one t-region, and precisely one overlay region.
Whenever a fragment is created, we issue a query to the {\sc Static Trellised} instance in its t-region, to determine its value.
We associate the initial value of each fragment to an arbitrary point within the fragment, and store these values in a data structure (which we call the \emph{fragment data structure}) that supports updates which add or remove a point in $\mathbb{Z}^d$ and range queries which ask for the value of $q$ evaluated over points in a range.
Such a data structure can be constructed with a $d$-dimensional segment tree.
After we have determined the value of the new fragments formed by a split, we remove the point corresponding to the fragment being split from the fragment data structure, then add to it new points located arbitrarily within the newly created fragments, with values corresponding to the respective fragments.
After the fragment data structure is updated, we can issue a range query over it to find the value of each newly created overlay region.
Each point update and range query over the data structure takes $O(\log^d s) = O(\log^d n)$ time, since $d$ is a constant and $s$ is bounded by a polynomial in $n$.

This concludes the description of the first part of the algorithm.
Before proceeding, let us pause to consider the time complexity of the procedure thus far, for any given batch.
%Our construction of the {\sc Static Trellised \problemname{$\setf$}{$+$}} instance for all t-regions takes $\TO(n^{3/2})$ time overall, with the remainder of the running time dominated by the queries to these instances.
There are initially $O(n^{d/2})$ t-regions, and thus, initial fragments.
Whenever a new coordinate is added to $C$, all fragments intersecting the hyperplane $H$ corresponding to this coordinate are split in two.
This creates a number of new fragments proportional to the number of such intersections, which in turn is proportional to the number of intersections between \emph{fragment boundaries} and $H$.
This intersection can be either at a boundary of an initial fragment (at most $O(n^{(d-1)/2})$ of these, by the second property of \autoref{overmars-yap}), or a boundary of an overlay region (at most $O(\bar{k}^{d-1})$ of these, since $H$ cannot intersect a boundary parallel to $H$).
Hence, the total number of fragments created over the batch is $f = O(n^{d/2} + \bar{k}(\bar{k}^{d-1} + n^{(d-1)/2}))$.
Using an identical argument to that in the proof of \autoref{minmax-reduction}, the time complexity of finding the initial value of these fragments is $O(n^{(c+d)/2} + fn^{(c-1)/2})$.

\subparagraph*{Supporting updates and queries.}
We now use the initial values of overlay regions to support updates and queries.
Let $\mathcal{O'} = \{ R_1 \times \ldots \times R_{d-1} : (R_1, \ldots, R_{d-1}) \in \mathcal{R}_1 \times \ldots \times \mathcal{R}_{d-1} \}$ be the overlay regions with respect to the first $d-1$ dimensions.
For each region in $\mathcal{O'}$, we will maintain an instance of {\sc 1D Partitioned \range{$U$}{$q$}}, with $C_{d}$ as the sequence of indices.
Recall that we add the coordinates of each update or query range to $C$, so we may assume each operation affects precisely a range of overlay regions.
To answer a query, we iterate over each region in $\mathcal{O'}$ in the query range, and query its corresponding {\sc Partitioned} instance, combining results along the way.

When performing an update, we assume that the initial values of each of the current overlay regions have already been computed by the first part of the algorithm.
Each coordinate added to $C_1$ through $C_{d-1}$ splits $O(\bar{k}^{d-2})$ regions in $\mathcal{O'}$ into two new regions.
For each, we create a new instance of {\sc 1D Partitioned \range{$U$}{$q$}} so after each, we have an instance for each region in the new $\mathcal{O'}$.
For each such region $O'$, we iterate through the updates, in order, in the batch so far (excluding the present update).
For each such update, its range is necessarily either disjoint from $O' \times \mathbb{Z}$, or its intersection with $O'$ is of the form $O' \times R_d$ for some range $R_d$ of the $d$-th dimension.
In the latter case, we apply the update to $R_d$.
Next, for each coordinate $x_d$ that has been added to dimension $d$, we iterate through all $O(\bar{k}^{d-1})$ {\sc Partitioned} instances, splitting each with $x_d$.
As a result, every range in every {\sc Partitioned} instance once again corresponds one-to-one with an overlay region.
For each newly created overlay region, we use its value to update the corresponding range on the corresponding {\sc Partitioned} instance.
Finally, we apply the present update to all {\sc Partitioned} instances.

To complete our proof of \autoref{thmt@@onlinemultidstatictrellis}, it remains to analyse the overall time complexity of our algorithm.
Recall that the number of dimensions, $d \geq 1$, is a constant.
We create $O(\bar{k}^{d-1})$ overlay regions overall.
Each is affected by $O(k)$ queries, updates, and splits, so we expend $\TO(\bar{k}^d)$ time on this data structure over all queries and updates in a batch.
Hence, the overall running time of the algorithm is dominated by the time taken to find the initial value of each fragment.
For balance, we set $\bar{k} = n^{(d+1)/2d}$, so $k = (\bar{k}-1)/2$ and $f = O(n^{(d+1)/2})$, since $d \geq 1$.
This gives an overall running time of $\TO(n^{\frac{c+d+1}{2} - \frac{1}{2d}})$ over the $n/k$ batches, as required.
This completes the proof of \autoref{thmt@@onlinemultidstatictrellis}.

By extending the {\sc Static Trellised} algorithm of \autoref{minmax-st} to multiple dimensions, we can apply this result to generalise our algorithm for \problemname{$\{\min, \setf, \max\}$}{$\max$} from \autoref{minmax-2d} to higher dimensions.

\begin{restatable}{theorem}{multidminsetmaxmaxthm}
  \label{multid-min-set-max-max}
  {\sc Grid} \range{$\{\min, \setf, \max\}$}{$\max$} can be solved in $\TO(n^{(d^2+2d-1)/2d})$ time.
\end{restatable}

\begin{proof}
    The first condition of \autoref{thmt@@onlinemultidstatictrellis} is met with $\bar{U} = \{ f_{a, b} : |a|, |b| \leq g \}$, where $g$ is the largest absolute value in the input, and $f_{a, b}(x) = \min(a, \max(b, x))$ as described earlier, which is closed under composition.
    The second condition is met using a {\sc 1D \range{$\{\min, \setf, \max\}$}{$\max$}} data structure.
    For each range of $A_0$, the data structure maintains a point placed arbitrarily within it.
    By using an ancillary lazy update data structure (as described in \autoref{subsec:sketch}), we can determine the composition of all updates to a given point.
    When given the maximum value $v$ of a new range of $A_0$, we are able to $\setf$ its corresponding point's value to the value obtained after all updates have been applied to $v$.
    Updates can then proceed as usual.
    Finally, it remains to solve {\sc $d$D Static Trellised Grid \range{$\{\min, \setf, \max\}$}{$\max$}} efficiently.

    To do so, we extend the {\sc Static Trellised} algorithm of \autoref{minmax-st} to multiple dimensions.
    Let $L_{i,j}$ be the latest label (as defined in \autoref{minmax-st}) covering coordinate $j$ in dimension $i$.
    We seek a point $(x_1, \ldots, x_d)$ for which there is a dimension $j$ such that $L_{j, x_j} + L_{j', x_{j'}} > 0$ for all $j' \neq j$.
    Hence, it suffices to find the latest label $L^{\ast}_i$ in every dimension $i$ and check if the sum of the latest and earliest among these is positive.
    This can be done in $\TO(1)$ time per update and $\TO(d) = \TO(1)$ time per query, by keeping a \autoref{buckets} data structure for each dimension.
    Hence, for constant $d$, {\sc Static Trellised \range{$\{\min, \setf, \max\}$}{$\max$}} can be solved in $\TO(n)$ time, and the result follows.
\end{proof}

In the next section, we give two additional applications of this theorem in two dimensions, which establish relationships between counting paths on graphs and {\sc Static Trellised} problems with $\setf$ updates and sum queries.

\section{Truly subquadratic \texorpdfstring{$\setf$}{set} updates and \texorpdfstring{$+$}{+} queries by counting paths}
\label{sec:setplus}

In this section, we apply \autoref{general-static-trellis} to give truly subquadratic algorithms for \setsum{} and \problemname{$\{+, \setf\}$}{$+$}.

\subsection{\texorpdfstring{\setsum{}}{2D Grid Range (+, set)} by counting inversions}
\label{sec:setsum}
The first condition of \autoref{general-static-trellis} is met for $U = \bar{U} = \{\setf\}$, since $\setf_a \circ \setf_b = \setf_a$ for any $a$ and $b$.
The second condition can be met solving {\sc 1D Partitioned Range $(+, \setf)$} with a data structure for {\sc 1D Grid Range $(+, \setf)$} over an array of size $n$.
To initialise the structure, for each $a_i$, we use an update to set the value at $a_i$ to $q([a_i, a_{i+1}-1]) = A[a_i] + \ldots + A[a_{i+1}-1]$.
All other values are set to 0.
Subsequent updates and queries are then passed through to the underlying data structure, which each occur in $O(\log n) = \TO(1)$ time.
Operations each occur in $O(\log n) = \TO(1)$ time.

Finally, we address the third condition by drawing an equivalence between {\sc Static Trellised \setsum} and a class of range query problems over arrays.
The {\sc RangeEqPairsQuery} accepts an array of size $n$ as input, and asks for the number of pairs of equal elements within each of $q$ given ranges.
Duraj et al. \cite{Duraj2020} defined this weighted analogue for counting inversions between pairs of ranges.

\begin{definition}[{\sc Weighted 2RangeInversionsQuery}]
    Given an integer array $A$, an integer array of weights $W$, both of length $n$, and a sequence of $q$ pairs of non-overlapping ranges $([l_1', r_1'], [l_1'', r_1'']), \dots, ([l_q', r_q'], [l_q'', r_q''])$, with $r_i' < l_i''$, compute for each pair $([l', r'], [l'', r''])$ the quantity
    \[
        \sum_{i \in [l', r']} \sum_{j \in [l'', r'']} \ind_{A[i] > A[j]} \cdot W[i] \cdot W[j].
    \]
    {\sc 2RangeInversionsQuery} is the problem with the added restriction that every weight is 1.
\end{definition}

\noindent
They showed that {\sc RangeEqPairsQuery} is equivalent, up to polylogarithmic factors, to {\sc 2RangeInversionsQuery}, even when the time complexity is expressed as a function of both $n$ and $q$.
We extend this equivalence to {\sc Static Trellised \setsum}.

\begin{lemma}
    \label{trellis-range-equiv}
    {\sc Static Trellised \problemname{$\setf$}{$+$}}, {\sc Weighted 2RangeInversionsQuery} and {\sc 2RangeInversionsQuery} all have the same complexity, up to polylogarithmic factors.
    This holds even when the queries are presented online, and with the complexity measured as a function of two variables, $n$ and $q$.
\end{lemma}

\begin{proof}
    ({\sc Static Trellised \problemname{$\setf$}{$+$}} $\rightarrow$ {\sc Weighted 2RangeInversionsQuery}).
    By adding a dummy update with value 0 covering the whole grid at time $-1$, we may assume -- without loss of generality -- that every row (column) is covered by at least one row (column) update. 
    Now observe that the value of a given point (after all updates) is equal to the value of the later of the latest row update and the latest column update covering this point's row and column, respectively.
    Thus, for each row (column), it suffices to keep the latest row (column) update covering it.
    With a sort and sweep, we can process the row updates into $O(n)$ disjoint row updates which preserve this property.
    Hence, without loss of generality, we may assume that the row updates are pairwise disjoint, as are the column updates.
    This preprocessing takes $O(n \log n)$ time.

    Consider a query over the points $[l_1, r_1] \times [l_2, r_2]$.
    We will show how to compute the contribution (sum of values) of those points whose value is determined by a row update: we can treat the contribution from column updates identically, and the sum of these contributions is the answer to the query.
    To do so, we form an instance of {\sc Weighted 2RangeInversionsQuery} as follows.
    Create an array $A$, whose values are equal to the update time of each row update in order, from top to bottom, concatenated with the update time of each column update in order, from left to right.
    For the weights of our instance, let the weight corresponding to the row update with value $v$ over rows $[l_1, r_1]$ be $v \times (r_1-l_1+1)$, and the weight corresponding to a column update over columns $[l_2, r_2]$ be $(r_2-l_2+1)$, irrespective of its value.
    This describes our instance (see \autoref{fig:reduce}).
    We will now describe the queries made to this instance.

    Since row (column) updates are disjoint, there is a contiguous range of row (column) updates in this array which lie entirely inside $[l_1, r_1]$ ($[l_2, r_2]$), whose indices can be found with binary search.
    The contribution of row updates to the points from these ranges is the result of a query on our instance over the corresponding ranges.
    $O(1)$ parts of the query range may fall within part of a row or column update: the contribution from these can be found by scaling the result of a similarly constructed query. % proportionate to the fraction of a row or column update included in the query.

    ({\sc Weighted 2RangeInversionsQuery} $\rightarrow$ {\sc 2RangeInversionsQuery}).
    We apply the same trick as Duraj et al. used to reduce multiedge instances of {\sc Triangle Counting} to single edge instances, with a multiplicative factor of $O(\log^2 n)$.
    This is achieved by converting each weight into a binary string of $\log n$ bits, and considering each possible pair of bit positions of the weights corresponding to the inverted elements.
    Formally, for an integer $b$ in $[0, \log n]$, let $A_b$ be the array $A$ containing only those $A[i]$ such that $W[i]$ has its $b$-th bit set.
    Then for each $(b_i, b_j)$ in $[0, \log n]^2$, we create an instance with $A_{b_i}$ concatenated with $A_{b_j}$.
    For each pair of query ranges $[l', r']$ and $[l'', r'']$, we query for $[l', r']$ in the first part of the array, and for $[l'', r'']$ in the second part of the array, adjusting the indices to account for the elements removed from the array.

    ({\sc 2RangeInversionsQuery} $\rightarrow$ {\sc Static Trellised \problemname{$\setf$}{$+$}}). 
    We produce an instance over an $n \times n$ grid, and perform (trellised) updates to the grid as follows.
    For each value $v$ appearing in $A$, in non-decreasing order: first, for each index $i$ such that $A[i] = v$, set the values in row $i$ to 1, then for each index $j$ such that $A[j] = v$, set the values in column $j$ to 0.
    The order guarantees that after these updates have been performed, $A[i][j] = 1$ if and only if $A[i] > A[j]$.
    Each pair of query ranges $[l', r']$ and $[l'', r'']$ corresponds to a query over the range $[l', r'] \times [l'', r'']$ in our instance.
\end{proof}

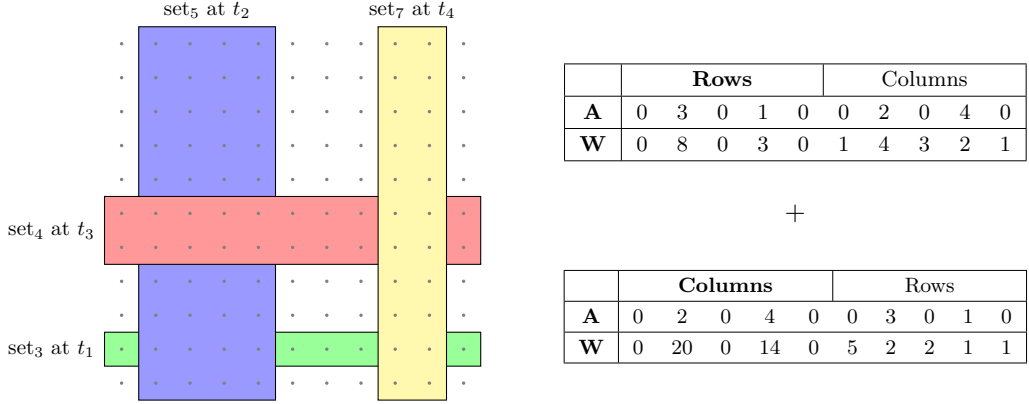
\begin{figure}
    \begin{tikzpicture}
        \begin{scope}[scale=0.45, every node/.style={scale=0.8}]
            \begin{scope}[shift={(0.3, 0)}]
                \draw[fill=green!40] (-0.5, 0.5) rectangle (10.5, 1.5);
                \draw[fill=blue!40] (0.5,-0.5) rectangle (4.5, 10.5);
                \draw[fill=red!40] (-0.5, 3.5) rectangle (10.5, 5.5);
                \draw[fill=yellow!40] (7.5,-0.5) rectangle (9.5, 10.5);
                \foreach \x in {0,1,...,10} {
                    \foreach \y in {0,1,...,10} {
                        \fill[color=gray] (\x,\y) circle (0.05);
                    }
                }
                \node[align=left] at (2.5, 11) {$\setf_5$ at $t_2$};
                \node[align=left] at (8.5, 11) {$\setf_7$ at $t_4$};
            \end{scope}
            \node[align=left] at (-1.75, 4.5) {$\setf_4$ at $t_3$};
            \node[align=left] at (-1.75, 1) {$\setf_3$ at $t_1$};
        \end{scope}

        \node[align=left] at (9, 2.25) {
                \resizebox{0.45 \textwidth}{!}{
                \begin{tabular}{|c|cccccccccc|}
                    \hline
                    & \multicolumn{5}{c|}{\bf Rows} & \multicolumn{5}{c|}{Columns} \\
                    \hline
                    {\bf A} & 0 & 3 & 0 & 1 & 0 & 0 & 2 & 0 & 4 & 0 \\
                    \hline
                    {\bf W} & 0 & 8 & 0 & 3 & 0 & 1 & 4 & 3 & 2 & 1 \\
                    \hline
                \end{tabular}
                }
                \\
                \\
                \\
                \\
                \\
                \resizebox{0.45 \textwidth}{!}{
                \begin{tabular}{|c|cccccccccc|}
                    \hline
                    & \multicolumn{5}{c|}{\bf Columns} & \multicolumn{5}{c|}{Rows} \\
                    \hline
                    {\bf A} & 0 & 2 & 0 & 4 & 0 & 0 & 3 & 0 & 1 & 0 \\
                    \hline                                           
                    {\bf W} & 0 & 20 & 0 & 14 & 0 & 5 & 2 & 2 & 1 & 1 \\
                    \hline
                \end{tabular}
                }
            };
        \node[align=center] at (9, 2.25) {$+$};
    \end{tikzpicture}
    \caption{Reducing {\sc Static Trellised \problemname{$\setf$}{$+$}} to {\sc Weighted 2RangeInversionsQuery}. Updates occur at times $t_1$ through $t_4$. Separate instances for row and column contributions.}
    \label{fig:reduce}
\end{figure}

\noindent
Duraj et al. \cite{Duraj2020} gave an $\TO(n^{(2\omega-2)/(\omega+1)}(n+q)^{2/(\omega+1)})$ time\footnote{The multivariate running time given in \cite{Duraj2020} is slightly better than this when $q \leq n$, but this simplified form suffices our purposes.} algorithm for {\sc RangeEqPairsQuery}, so we obtain the following truly subquadratic time algorithm for \setsum.

\begin{restatable}{theorem}{setsumthm}
    \label{setsum}
    \setsum{} can be solved in $\TO(n^{5/4 + \omega/(\omega+1)}) = O(n^{1.954})$ time. 
\end{restatable}

%\jl{Consider remarking why this doesn't generalise to higher dimensions.}

\subsection{\texorpdfstring{\problemname{$\{+, \setf\}$}{$+$}}{2D Grid Range (+, \{+, set\})} by counting 3-paths}
\label{sec:setplus-sum}

We will once again employ \autoref{general-static-trellis} to give a truly subquadratic time algorithm for \problemname{$\{+, \setf\}$}{$+$}.
The first condition is met with $\bar{U} = \{+_c: |c| \leq ng\} \dot\cup \{\setf_c: |c| \leq ng\} $, where $g$ is the largest value in the input, since $+_a \circ \setf_b = \setf_{a + b}$.
These can be represented in $O(1)$ words.
The second condition can be met with a {\sc 1D Grid \range{$\setf$}{$+$}} data structure, constructed similarly as in the previous section.
Fulfilling the third condition is the subject of the remainder of this section.

We begin by defining the following graph problems.

\begin{definition}
    The {\sc $k$-WalkQuery} (resp. {\sc Simple$k$-PathQuery}) problem gives, as input, a simple graph with $m$ edges and $O(m)$ vertices, and poses $q$ online queries, each asking for the number of $k$-edge walks (resp. simple $k$-edge paths) between a given pair of vertices.
\end{definition}

Duraj et al. \cite{Duraj2020} proved an equivalence between {\sc RangeEqPairsQuery} when $n = q$, and {\sc EdgeTriangleCounting}, which gives as input an $n$-edge, $O(n)$-vertex graph, and asks to count the number of triangles each edge is contained in.
When the restriction $n = q$ is relaxed, an equivalence can be drawn with {\sc 2WalkQuery} instead. %, which asks: given a graph, answer queries which ask for the number of 2-walks between a given pair of vertices.
In the same vein, we reduce {\sc Static Trellised \problemname{$\{+, \setf\}$}{$+$}} to {\sc 3WalkQuery} via a generalisation of {\sc RangeEqPairsQuery}.

% \begin{lemma}
%     If {\sc 3WalkQuery} can be solved in $T(m, q)$ time then {\sc Static Trellised \problemname{$\{+, \setf\}$}{$+$}} can be solved in $\TO(T(\numupdates, \numqueries))$ time.
% \end{lemma}

First, we define a class of problems that generalises {\sc 2RangeEqPairsQuery} and {\sc 2RangeInversionsQuery}.

\begin{definition}
    Let $A$ and $W_A$ be an integer arrays of length $n$, and suppose $L$, $R$ and $W$ are arrays of length $n$, such that $1 \leq L[i] \leq R[i] \leq n$ for all $i$.
    %Let $B$ be the array that is the concatenation of the subarrays $A[L[1], R[1]], \ldots, A[L[n], R[n]]$.
    The {\sc Weighted-$(0,1)$-IndexedEqPairsQuery} problem poses $q$ queries, each accepting a pair of ranges $[l, r]$ and $[l', r']$ and asking for the sum
    \[
        \sum_{i \in [l, r]} \sum_{k' \in [l', r']} \sum_{j \in [L[k'], R[k']]} \ind_{A[i] = A[j]} \cdot W_A[i] \cdot W_A[j] \cdot W[k'].
    \]
    The {\sc Weighted-$(0,1)$-IndexedInversionsQuery} problem poses $q$ queries, each accepting a pair of ranges $[l, r]$ and $[l', r']$ and asking for the sum
    \[
        \sum_{i \in [l, r]} \sum_{k' \in [l', r']} \sum_{j \in [L[k'], R[k']]} \ind_{A[i] > A[j]} \cdot W_A[i] \cdot W_A[j] \cdot W[k'].
    \]
    The {\sc $(0,1)$-IndexedEqPairsQuery} and {\sc $(0,1)$-IndexedInversionsQuery} are the problems with the added restriction that every entry in $W_A$ and $W$ is 1.
\end{definition}

These problems are so named as the first range in each query $[l, r]$ directly indexes the array $A$, whereas the second range $[l', r']$ indexes an array of ranges (defined by $L$ and $R$).
Hence, the parameters undergo 0 and 1 levels of indirection, respectively.
{\sc 2RangeEqPairsQuery} and {\sc 2RangeInversionsQuery} are thus {\sc $(0,0)$-IndexedEqPairsQuery} and {\sc $(0,0)$-IndexedInversionsQuery}, respectively. 

Through reductions similar to Duraj et al. \cite{Duraj2020} and those in the previous subsection, these are equivalent, up to polylogarithmic factors.

\begin{lemma}
    \sloppy
    {\sc $(0,1)$-IndexedEqPairsQuery}, {\sc Weighted-$(0,1)$-IndexedEqPairsQuery}, {\sc $(0,1)$-IndexedEqPairsQuery} and {\sc Weighted-$(0,1)$-IndexedEqPairsQuery} have the same running time, up to polylogarithmic factors, even when the running time is expressed as a function of $n$ and $q$.
\end{lemma}

\begin{proof}
%    This closely resembles the proof of Duraj et al. \cite{Duraj2020}.
%
    ({\sc $(0,1)$-IndexedEqPairsQuery} $\rightarrow$ {\sc $(0,1)$-IndexedInversionsQuery})
    We can easily compute the number of pairs of entries using prefix sums: this takes time linear in the input, and any algorithm for {\sc $(0,1)$-IndexedInversionsQuery} takes at least as long.
    For each query, count the number of inversions between the pair of query ranges.
    By maintaining a copy of $A$ with all entries negated, we can also count the number of inversions in the opposite direction.
    Subtracting the answers to these two queries from the number of pairs in each query gives the answer.

    ({\sc $(0,1)$-IndexedInversionsQuery} $\rightarrow$ {\sc $(0,1)$-IndexedEqPairsQuery})
    Suppose each of the entries of $A$ are $z$-bit integers.
    If $a > b$, then the (padded) binary representations of $a$ and $b$ share some prefix, with $a$'s next bit being 1, and $b$'s next bit being 0.
    We reduce to $O(\log n)$ instances of {\sc $(0,1)$-IndexedEqPairsQuery}: one for each prefix.
    Let $v_i(x)$ be the $i$th bit of $x$, from most significant, and let $p_i(x) = \lfloor x / 2^{z-i} \rfloor$ be the $i$ most significant bits of $x$.
    For each $i = 1, \ldots, z$ we create an array $A_i$ such that:
    \begin{align*}
        A_i[j] &= p_{i-1}(A[j]) \text{ if $v_i(A[j]) = 1$ and } -\infty \text{ otherwise} \\
        A_i[n + j] &= p_{i-1}(A[j]) \text{ if $v_i(A[j]) = 0$ and } \infty \text{ otherwise}
    \end{align*}
    The answer to a query over $A$ is the sum of the answers to the queries over each $A_i$.
\end{proof}

Next, we reduce from {\sc Static Trellised \problemname{$\{+, \setf\}$}{$+$}} to these problems.

\begin{lemma}
    If {\sc $(0,1)$-WeightedIndexedInversionsQuery} can be solved in $T(n, q)$ time, then {\sc Static Trellised \problemname{$\{+, \setf\}$}{$+$}} can be solved in $\TO(T(n, q))$ time.
\end{lemma}

\begin{proof}
    %We reduce from {\sc Static Trellised \problemname{$\{+, \setf\}$}{$+$}}.
    Observe that the value of any point is equal to its value at the time $t$ of the latest $\setf$ update covering it, plus the sum of all $+$ updates covering it that occur after $t$.
    Hence, we can consider the contributions of $\setf$ and $+$ updates separately.
    The contributions of $\setf$ updates can be computed in the same way as for {\sc Static Trellised \problemname{$\setf$}{$+$}}: namely, by reducing to {\sc 2RangeInversionsQuery} which is the special case of {\sc $(0,1)$-IndexedInversionsQuery} with all $L[i] = R[i] = i$.

    It remains to determine the contribution of the $+$ updates.
    We will show how to determine the contributions of row $+$ updates only: column $+$ updates can be determined separately in the same way.
    We construct an instance of {\sc Weighted $(0,1)$-IndexedInversionsQuery}, over an array $A = A_1 A_2$.
    By compressing the coordinates of the grid according to the coordinates present in $+$ and $\setf$ updates, we assume there are $n' \leq 2n+1$ \emph{compressed} rows and columns, each with a certain width.
    Every update covers precisely a range of compressed rows and compressed columns.

    Let $t_R[i]$ be the time of the latest $\setf$ row update covering compressed row $i$, and let $t_C[j]$ be the time of the latest $\setf$ column update covering compressed column $j$.
    $A_1$ will have length $n'$, with $A_1[i] = t_C[i]$, and the weight of the $i$th entry will be the width of compressed column $i$.
    Next, form a segment tree over compressed rows.
    For each range $[l, r]$, store a list of all row $+$ updates that decompose to $[l, r]$ (as defined in \autoref{sec:prelim}), sorted in decreasing order of time.
    Formally, the updates in this list are precisely those of the form ``apply $+_c$ to $[l_1, r_1] \times \mathbb{Z}$'', where $[l, r] \in \text{base}([l_1, r_1])$.
    Build an array of size $O(n \log n)$ by concatenating these lists in any order, while preserving the order within each list, with the value of each entry equal to the time of the entry.
    The weight of each entry is equal to the value of the $+$ update.
    This forms $A_2$.

    To build $L$ and $R$, we build two smaller lists, $L_i$ and $R_i$, for each compressed row $i$, and concatenate them in order.
    For compressed row $i$, for each segment tree range $i$ is in, binary search to find the largest prefix of entries in that range that occur after $t_R[i]$. %\ar{inconsistent capital/lowercase T}
    We append entries in $L_i$ and $R_i$ corresponding to the range of $A_2$ that these entries lie in.
    The weight of each entry is equal to the width of compressed row $i$.

    Each query corresponds to a range of rows and a range of columns.
    We decompose the range of rows into a range of whole compressed rows, and at most two fractions of compressed rows, and do the same for columns.
    Thus, we can write the query answer as the sum of queries over ranges of entire compressed rows and columns, scaling the result of a query over a fraction of a compressed row and/or column.
    For each of these ranges, we form a query by setting $[l, r]$ to correspond to the range of compressed columns in $A_1$, and set $[l', r']$ to be the range of entries in $L$ and $R$ corresponding to the range of compressed rows.
    %To account for weights, we try all $O(\log^3 n)$ triples of bit positions over the 3 weight arrays, and sum them up.
\end{proof}

As a means to solve {\sc $(0,1)$-IndexedEqPairsQuery}, we prove an equivalence between it and {\sc 3WalkQuery}.
Queries to {\sc 3WalkQuery} must be answered in an online fashion to give an online data structure for our {\sc Static Trellised} problem.

%The {\sc Simple3PathQuery} problem gives, as input, a graph with $n$ edges, $O(m)$ vertices, and poses $q$ queries, each asking for the number of simple paths of 3 edges between a given pair of vertices.
%This generalises {\sc Edge4CycleCounting}, which asks for the number of simple 4-cycles each edge is contained in, as each instance can be reduced to a {\sc Simple3PathQuery} instance with a query for each edge.

\begin{lemma}
    {\sc $(0,1)$-IndexedEqPairsQuery} is equivalent to {\sc 3WalkQuery}, up to polylogarithmic factors, even when the running time is expressed as a function of both $n$ and $q$.
\end{lemma}

\jl{Consider a diagram for this proof}
\begin{proof}
%    Let {\sc Partitioned 2RRREqPairsQuery} be the {\sc 2RRREqPairsQuery} problem, with the restriction that query ranges correspond to ranges of entire subarrays.
%    We will reduce from {\sc 2RRREqPairsQuery} to {\sc 3WalkQuery} via this restricted problem.
%
%    ({\sc 2RRREqPairsQuery} $\rightarrow$ {\sc Partitioned 2RRREqPairsQuery})
%    Suppose $A, L, R$ are arrays of length $n$ that form a {\sc 2RRREqPairsQuery} instance.
%    We construct an array $A'$ of length $2n$ by appending arbitrary elements to $A$.
%    We let $L'$ and $R'$ be arrays of length $2n$, such that $L'[1, n] = L$, $R'[1, n] = R$, and $L'[n+i] = R'[n+i] = i$ for each $1 \leq i \leq n$.
%    We form an instance of {\sc Partitioned 2RRREqPairsQuery} over $(A', L', R')$, denoting by $B'$ the array being queried.
%
%    Now any range of $B$ is either a range of a subarray, or can be decomposed into three disjoint ranges: a prefix of a subarray, a suffix of a subarray and a range which covers precisely a range of entire subarrays.
%    Any range of a subarray in $B$ corresponds to a range of $A$, which can be found as a range of entire subarrays at the end of $B'$.
%    Hence, the answer to any query in the original instance can be written as the sum of at most 3 queries in the restricted instance.

    ({\sc $(0,1)$-IndexedEqPairsQuery} $\rightarrow$ {\sc 3WalkQuery}) %\ar{inconsistent range tree vs segment tree}
    We construct a graph whose vertex set is $S \dot\cup T_1 \dot\cup U \dot\cup T_2$, and define these four sets as follows.
    $S$ is the set of nodes in a segment tree with $n$ leaves, denoting the indices of $L$ and $R$.
    $T_1$ and $T_2$ are the set of nodes in a segment tree with $n$ leaves, denoting the $n$ indices of $A$.
    $U$ is the set of distinct values in $A$.
    For the node $s$ representing range $[l, r]$ in $S$, for each $i \in [l, r]$ and for each $I \in \text{base}([L[i], R[i]])$, add an edge from $s$ to the node in $T_1$ corresponding to $I$.
    There are $O(n \log n)$ such $(s, i)$ pairs, and $O(\log n)$ elements in $\text{base}([L[i], R[i]])$ interval, for $O(n \log^2 n)$ edges added.

    Next, for the node $t$ representing range $[l, r]$ in $T_1$, for each $i \in [l, r]$ add an edge from $t$ to the node in $U$ representing the value $A[i]$.
    Do likewise for $T_2$.

    In all, the graph contains $O(n)$ vertices and $O(n \log^2 n)$ edges.
    Every 3-edge walk from $S$ to $T_2$ passes through $T_1$ and $U$.
    To answer a query $[l, r]$ and $[l', r']$, for each pair of ranges $(I, I') \in \text{base}([l, r]) \times \text{base}([l', r'])$, add the number of 3-edge walks between the vertex corresponding to $I$ in $S$ and the vertex corresponding to $I'$ in $T_2$.

    ({\sc 3WalkQuery} $\rightarrow$ {\sc $(0,1)$-IndexedEqPairsQuery})
    Let $A$ be the concatenation of the adjacency lists of each vertex, in order.
    Let $L[i]$ and $R[i]$ correspond to the range in $A$ where the $i$th vertex's adjacency list is located.
    For each pair of query vertices $s$ and $t$, we add a query for the ranges $[L[s], R[s]]$ and $[L[t], R[t]]$.
    The equal elements correspond to vertices that can be reached in 1 edge from $s$, and in 2 edges from $t$.
\end{proof}

%If each query asks for $E_A([l, r], [l', r'])$ directly, this is proves the equivalence between between {\sc 2RangeEqPairsQuery} and {\sc 2WalkQuery}, matching the construction given by Duraj et al. \cite{Duraj2020}.
The structure of {\sc $(0,1)$-IndexedEqPairsQuery} and {\sc $(0,1)$-IndexedInversionsQuery} can be generalised by adding additional layers of ``indexing'' arrays $L$ and $R$ to both the first and second ranges in each query.
Informally, we remark that if there are $p$ layers of indexing applied to the first range, and $q$ layers to the second range, then we obtain problems equivalent to {\sc $(p+q+2)$-WalkQuery}, up to polylogarithmic factors of order $p+q+O(1)$.

%This proof can be generalised to show that {\sc $(a,b)$-IndexedEqPairsQuery} is equivalent to {\sc $(a+b+2)$-WalkQuery}, up to polylogarithmic factors.
%Finally, we show that {\sc 3WalkQuery} can be solved efficiently.
%
% To solve {\sc 3WalkQuery}, we generalise the 4-cycle detection and counting algorithms of Yuster and Zwick \cite{Yuster2004} and Vassilevska Williams et al. \cite{Williams2015} to solve {\sc 3WalkQuery}.
% Specifically, we partition vertices into three groups based on their degree, and consider each of the possible configurations of vertices in the 3-Walk with respect to these groups.
% For each configuration, we use a combination of rectangular matrix multiplication and enumerating edges, both during precomputation and on-the-fly for each query.
% Finally, we perform a multivariate analysis of the running time, obtaining the following result.

\begin{theorem}
    \label{thm:threewalk}
    {\sc 3WalkQuery} can be solved in $O(m^{2\omega/(2\omega+1)} (m+q)^{(2\omega-1)/(2\omega+1)})$ time.
    % $O(m^{(4\omega-1)/(2\omega+1)})$ query-subadditive time. This gives a running time of , and thus an $O(m^{(4\omega-1)/(2\omega+1)})$ query-subadditive running time overall.
\end{theorem}

\begin{proof}
    Let $\Delta > 0$ be a constant, and partition vertices into three disjoint sets $H$, $M$, and $L$ depending if their degree is greater than $\Delta$, greater than $\sqrt{\Delta}$ but at most $\Delta$, and at most $\sqrt{\Delta}$, respectively.
    Next, enumerate all 2-walks $abc$ such that $b \in M \cup L$, and store a count $C(a, c)$ of the number of such walks for each $(a, c)$ pair (note that possibly $a = c$).
    There are at most $m\Delta$ of these walks.
    From this, we can construct matrices $C_{H, H}$, $C_{H, M, M}$ and $C_{M, L, H}$, which are the counts with the restrictions $a, c \in H$; $a \in H$ and $b, c \in M$; and $a \in M, b \in L$ and $c \in H$, respectively.
    This takes time $O(m\Delta + |H||M|)$.
    We also compute the $|H| \times |M|$ matrix $D = G[H] \times G[H, M]$ in time $O(\omega(2m/\Delta, 2m/\Delta, 2m/\sqrt{\Delta}))$, where -- through a slight abuse of notation -- we write $G[H]$ for the adjacency matrix of the subgraph of $G$ induced by the vertex set $H$, and write $G[H, M]$ for the adjacency matrix of the bipartite graph induced by the disjoint vertex sets $H$ and $M$ in $G$.

    We consider disjoint cases according to which set each of the four vertices in the 3-walk falls into.
    For ease of description, we describe them as if they formed a closed 4-walk, by imagining an edge joining the two endpoints (see \autoref{fig:fourcycle}).
    For each, we describe any necessary precomputation, and how to answer each query.
    For each query between vertices $u$ and $v$, we consider the sets $u$ and $v$ are in, and sum the number of 3-walks in each case to obtain the result.

    \subparagraph*{Case 1: All vertices in $H$.}
    We directly precompute the number of 3-walks between each pair of vertices in $G[H]$, using matrix multiplication.
    This takes time $O(|H|^{\omega}) = O((2m/\Delta)^{\omega})$.

    \subparagraph*{Case 2: Two opposite vertices in $M \cup L$.}
    Consider a pair of query vertices $u$ and $v$, with at least one in $M \cup L$.
    Without loss of generality, if $u \in M \cup L$ but $v \not\in M \cup L$ then, the answer is $P(u, v) = \sum_{ua \in E} C(a, v)$.
    Otherwise, if both $u$ and $v$ are in $M \cup L$, the answer is $P(u, v) + P(v, u) - \sum_{ua \in E: a \in M \cup L} C(a, v)$.
    Each of these sums can be computed in $O(\Delta)$ time, since $d(u) \leq \Delta$ whenever $u \in M \cup L$.
    Hence, this case takes $O((m + q)\Delta)$ over all queries.

    \subparagraph*{Case 3: Three vertices in $H$.}
    We can compute the number of walks of the form $H \rightarrow M \cup L \rightarrow H \rightarrow H$ by multiplying $C_{H, H}$ by the adjacency matrix of $G[H]$ to give the matrix $P$, taking time $O((2m/\Delta)^\omega + m\Delta)$.
    Let $h_i$ be the $i$th $H$ vertex.
    If $u$ and $v$ are both in $H$, the answer is $P_{uv} + P_{vu}$.

    Otherwise, without loss of generality, assume $u \in H$ and $v \in M \cup L$.
    We can compute the number of 2-walks between pairs of vertices in $H$ by squaring the adjacency matrix of $G[H]$.
    We can then iterate over each vertex $h$ in $H$ adjacent to $v$, and add the number of 2-walks between $u$ and $h$.
    Hence, this case takes $O((m/\Delta)^\omega + (m + q)\Delta)$ time over all queries.

    \subparagraph*{Case 4a: Two adjacent vertices in $H$, two adjacent vertices in $M$.}
    We multiply $C_{H, M, M}$ by the adjacency matrix of $G[M, H]$ to compute the number of walks of the form $H \rightarrow M \rightarrow M \rightarrow H$ between all pairs of vertices in $H$.
    This takes time $O(\omega(2m/\Delta, 2m/\sqrt{\Delta}, 2m/\Delta))$.
    If $u, v \in H$, we can read off the answer from the matrix.

    Otherwise, at least one of $u$ and $v$ is in $M$: suppose, without loss of generality, that $u$ is.
    If $v \in H$, then the answer can be computed as $\sum_{um \in E : m \in M} D_{vm}$.
    If $v \in M$, we can find the answer as $\sum_{uh \in E : h \in H} D_{hv}$.
    Both cases can be computed in time $O(q\Delta)$ over all queries, since $d(u) \leq \Delta$.

    \subparagraph*{Case 4b: Two adjacent vertices in $H$, one vertex in $M$, one vertex in $L$.}
    In precomputation, multiply $G[H, M]$ by $C_{M, L, H}$ in time $O(\omega(2m/\Delta, 2m/\sqrt{\Delta}, 2m/\Delta))$ to compute the number of walks of the form $H \rightarrow M \rightarrow L \rightarrow H$ between all pairs of vertices in $H$.
    If this forms a matrix $F$, the answer when $u, v \in H$ is $F_{uv} + F_{vu}$.
    We also precompute for the case when $u \in M$ and $v \in H$ by multiplying $C_{M, L, H}$ by the adjacency matrix of $G[H]$.
    In this case, the answer can be read off directly.

    Otherwise, at least one of $u$ and $v$ is in $L$: suppose, without loss of generality, that $u$ is.
    If $v \in H$, we can compute the answer as $\sum_{um \in E : m \in M} D_{vm}$.
    Otherwise, $v \in M$ and we can compute the answer as $\sum_{uh \in E : h \in H} D_{hv}$.
    This runs in time $O(q \sqrt{\Delta})$ over all queries, since $d(u) \leq \sqrt{\Delta}$.

    \subparagraph*{Case 4c: Two adjacent vertices in $H$, two adjacent vertices in $L$.}
    First, we precompute the answer for the case where $u, v \in H$.
    For each $u, v \in H$, initialise a count $Q_{uv}$ to 0.
    For each edge $ab \in L$, we iterate over each pair of edges $ua$ and $bv$, such that $u, v \in H$, and increment $Q_{uv}$.
    This takes time $O(m \Delta)$ overall, since each $d(a), d(b) \leq \sqrt{\Delta}$.
    We can then read the answer for each query off $Q$.

    Otherwise, at least one of $u$ and $v$ is in $L$: suppose, without loss of generality, that $u$ is.
    If $v \in H$, then we can compute the answer as $\sum_{ul \in E : l \in L} |\{lh \in E : h \in H, hv \in E\}|$ in time $O(\Delta)$ per query, since $d(u), d(l) \leq \sqrt{\Delta}$.
    Otherwise, $v \in L$, and the answer can be computed as $\sum_{ua \in E : a \in H} |\{bv \in E : b \in H, ab \in E\}|$ in time $O(\Delta)$ per query, since $d(u), d(v) \leq \sqrt{\Delta}$.
    Hence, both cases can be computed in time $O(q\Delta)$ overall.

\begin{figure}
    \begin{tikzpicture}
        \tikzset{
            pics/fourcycle/.style args={#1,#2,#3,#4}{
                code={
                    \node (A)[align=center] at (1, 0) {#1};
                    \node (B)[align=center] at (2, -1) {#2};
                    \node (C)[align=center] at (1, -2) {#3};
                    \node (D)[align=center] at (0, -1) {#4};
                    \draw (A) -- (B) -- (C) -- (D);
                    \draw [dashed] (D) -- (A);
                }
            },
            casenum/.style={above left=.5cm, draw, inner sep=1pt, minimum size=13pt}
        }

        \def\vgap{3.2}

        \node[casenum] at (0, 0) {\bf 1};
        \draw (0, 0) pic{fourcycle={$H$, $H$, $H$, $H$}};

        \node[casenum] at (4, 0) {\bf 2};
        \draw (4.25, 0) pic{fourcycle={$M \cup L$, $H \cup M \cup L$, $M \cup L$, $H$}};
        \draw (8.5, 0) pic{fourcycle={$M \cup L$, $H \cup M \cup L$, $M \cup L$, $M \cup L$}};

        \node[casenum] at (0, -\vgap) {\bf 3};
        \draw (0, -\vgap) pic{fourcycle={$H$, $M \cup L$, $H$, $H$}};
        \draw (4.25, -\vgap) pic{fourcycle={$H$, $H$, $H$, $M \cup L$}};

        \node[casenum] at (0, -2*\vgap) {\bf 4a};
        \draw (0, -2*\vgap) pic{fourcycle={$H$, $M$, $M$, $H$}};
        \draw (4.25, -2*\vgap) pic{fourcycle={$M$, $H$, $M$, $H$}};
        \draw (8.5, -2*\vgap) pic{fourcycle={$M$, $H$, $H$, $M$}};

        \node[casenum] at (0, -3*\vgap) {\bf 4b};
        \draw (0, -3*\vgap) pic{fourcycle={$H$, $M$, $L$, $H$}};
        \draw (3.5, -3*\vgap) pic{fourcycle={$M$, $L$, $H$, $H$}};
        \draw (7, -3*\vgap) pic{fourcycle={$L$, $M$, $H$, $H$}};
        \draw (10.5, -3*\vgap) pic{fourcycle={$L$, $H$, $H$, $M$}};

        \node[casenum] at (0, -4*\vgap) {\bf 4c};
        \draw (0, -4*\vgap) pic{fourcycle={$H$, $L$, $L$, $H$}};
        \draw (4.25, -4*\vgap) pic{fourcycle={$L$, $L$, $H$, $H$}};
        \draw (8.5, -4*\vgap) pic{fourcycle={$L$, $H$, $H$, $L$}};
    \end{tikzpicture}

    \caption{The cases in proving \autoref{thm:threewalk}, in order of consideration. In each case, the top vertex is $u$, and the left vertex is $v$, with the vertex pair $(u, v)$ given in the query connected by a dashed line in the diagram, irrespective of whether they are joined by an edge in the input graph. We omit cases which are identical up to reflection.}
    \label{fig:fourcycle}
\end{figure}
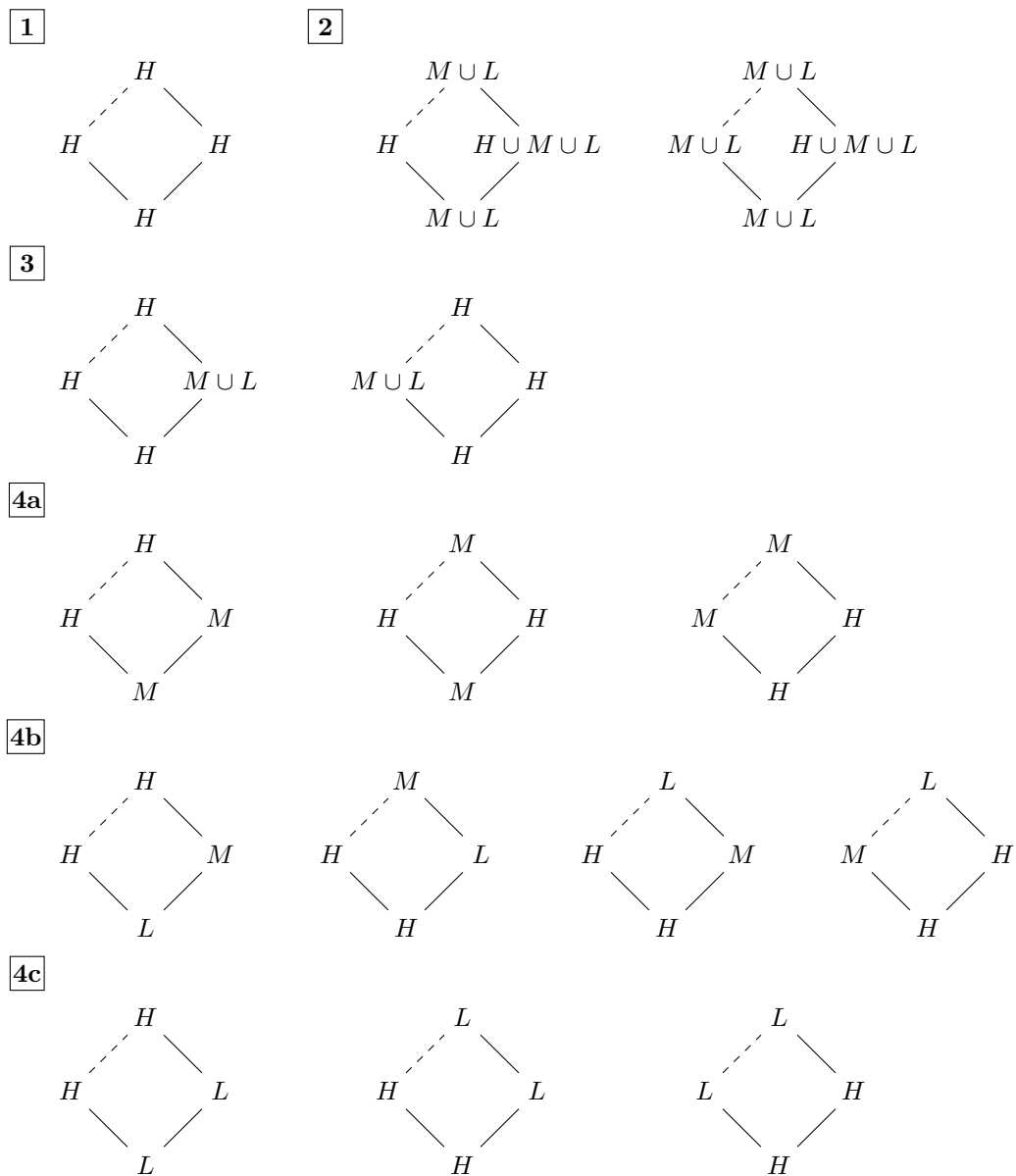

    \medskip
    \noindent
    This concludes the description of our algorithm; we now analyse it.
    The time bottlenecks in our algorithm are the $O((m + q)\Delta)$ steps, and the rectangular matrix multiplication steps.
    Let $\Delta = m^{\alpha}$, for some $\alpha \in (0, 1)$.
    Using the upper bound $\omega(n^a, n^b, n^c) = O(n^{a + b + c - (3 - \omega) \min \{a, b, c\}})$, the latter steps run in time $O(m^{3-5\alpha/2-(3-\omega)(1-\alpha)}) = O(m^{\alpha/2 + \omega(1-\alpha)})$.
    Hence, we obtain an overall time complexity of $O((m + q)m^{\alpha} + m^{\alpha/2 + \omega(1-\alpha)})$.

    When $q \leq m$, we set $\alpha = 2(\omega-1)/(2\omega+1)$ for a complexity of $O(m^{(4\omega-1)/(2\omega+1)})$ as obtained by Yuster and Zwick \cite{Yuster2004}.
    Otherwise, let $\beta = \frac{\log q}{\log m}$, and set $\alpha = \frac{2(\omega - \beta)}{2\omega + 1}$ for balance.
    We may assume $\beta \leq 2 \leq \omega$ if we cache the results to each query, since there are only $O(m^2)$ distinct queries.
    This gives a running time of $O(m^{2\omega/(2\omega+1)} q^{(2\omega-1)/(2\omega+1)})$ when $q > m$, and thus an $O(m^{2\omega/(2\omega+1)} (m+q)^{(2\omega-1)/(2\omega+1)})$ time algorithm overall.

    Finally, if the value of $q$ is not known during precomputation, we run the algorithm with an initial guess of $q = 1$, and double our guess when more than $q$ queries have been received.
    This increases the runtime by at most a constant multiplicative factor.
\end{proof}

This gives us an algorithm with the same running time (up to polylogarithmic factors) for {\sc Static Trellised \problemname{\{$+, \setf$\}}{$+$}}, and thus the unrestricted problem, by \autoref{general-static-trellis}.

\begin{restatable}{theorem}{setplussum}
    \label{set-plus-sum}
    \problemname{$\{+, \setf\}$}{$+$} can be solved in time $\TO(n^{5/4 + (4\omega-1)/(4\omega+2)})$ $= O(n^{1.989})$.
\end{restatable}

To conclude this section, we also remark that our algorithm for {\sc 3WalkQuery} can be used to solve the analogous problem for simple 3-edge paths, and thus count the number of 4-cycles each edge of a graph is included in.

\begin{theorem}
    {\sc 3WalkQuery} is equivalent to {\sc Simple3PathQuery}, and both can be solved in $O(m^{2\omega/(2\omega+1)} (m+q)^{(2\omega-1)/(2\omega+1)})$ time.
    Hence, {\sc Edge4CycleCounting} (which asks to count the number of simple 4-cycles that each edge in an $m$-edge graph is contained in) can be solved in $O(m^{(4\omega-1)/(2\omega+1)}) = O(m^{1.478})$ time.
\end{theorem}

\begin{proof}
    Let $u$ and $v$ be a pair of distinct query vertices, and denote by $P_3(u, v)$ and $W_3(u, v)$ the number of simple 3-edge paths, and the number of 3-edge walks between $u$ and $v$.
    Suppose $uabv$ is a non-simple 3-edge walk.
    Since a simple graph contains no self-loops, $a \neq u$, $a \neq b$ and $b \neq v$.
    But since $uabv$ is non-simple, some vertex must be repeated, so either $a = v$, $b = u$, or both.
    Hence, if the edge $uv$ exists in the graph, $W_3(u, v) = P_3(u, v) + d(u) + d(v) - 1$, where $d(u)$ is the degree of $u$.
    Otherwise, every 3-walk from $u$ to $v$ is also a simple 3-path, so $W_3(u, v) = P_3(u, v)$.

    \sloppy
    Since the edge $uv$, if it exists, cannot be part of any simple 3-path from $u$ to $v$ it follows that {\sc Edge4CycleCounting} can be solved by counting the number of simple 3-paths between the endpoints of each edge.
\end{proof}

\noindent
{\sc Simple-$k$PathQuery} is $\#W[1]$-hard for parameter $k$ \cite{Flum2004}, meaning it is unlikely that an $f(k) n^{O(1)}$ time algorithm exists for it.
However, {\sc $k$WalkQuery} is solvable in $O(n^{\omega} \log k)$ on $n$-vertex graphs by exponentiating the adjacency matrix, so it is unlikely that such an equivalence holds between the two problems for arbitrary $k$.

\section{{\sc Range} problems over an explicit point set}

\label{sec:pointwise}

\kdtree*

\begin{proof}
    We construct a $k$d-tree \cite{Bentley1975} $T$ over the points in $P$.
    There is a single point of $P$ in each of the leaves of $T$: we assign these labels from 1 to $p$, according to the order of their appearance in a preorder traversal of $T$.

    We can represent the points in any orthogonal range as those in the disjoint union of $O(p^{1-1/d})$ subtrees of $T$ \cite{Lee1977}.
    The labels of points within these subtrees correspond to disjoint ranges of $[p]$.
    Hence, we can keep an instance of {\sc 1D \range{$U$}{$q$}}, and perform updates and queries on the corresponding ranges of labels using this data structure. %\ar{notation used here: number of points = $n$}
\end{proof}

\begin{theorem}
   If either \pointwise{$+$}{$+$}{2} or \pointwise{$\max$}{$\max$}{2} can be solved in amortised $O(p^{1/2 -\epsilon})$ time per update or query, for any $\epsilon > 0$, then the \nameref{conj:OMv} is false.
\end{theorem}

%\begin{proof}[Proof sketch]
%    We use a slight modification of \autoref{omv-gridwise}, with a point in $P$ for each $M_{ij} = 1$.
%    For each pair of query vectors $(u, v)$, we use the update operation to mark row $i$ for each $u_i = 1$.
%    For each $v_j = 1$, we perform a query over column $j$ to check if any points in that column were marked.
%    A trick for \pointwise{$\max$}{$\max$}{2} allows us to reuse the same instance for all query vector pairs.
%\end{proof}

\begin{proof}
    We reduce from OuMv.
    For each $M_{ij} = 1$, create a point $(i, j)$ in our instance of the data structure. 
    
    (\pointwise{$+$}{$+$}{2}) 
    We maintain that at the beginning of each query, the value of each point is set to 0. 
    We perform the following for each pair of query vectors $u$ and $v$. 
    For each $u_i = 1$, add 1 to all points in row $i$ in the data structure.
    After doing this, use the data structure to query the sum of all points in all columns $j$ where $v_j = 1$. 
    We observe that $u^{T}Mv = 1$ if and only if at least one of these queries was nonzero. 
    After performing these queries, we revert all of the updates by subtracting 1 from each row that we updated.
    
    (\pointwise{$\max$}{$\max$}{2})
    We maintain that at the beginning of the $i$-th query, the value of each point is at most $i-1$.
    We proceed similarly to \pointwise{$+$}{$+$}{2}, performing updates with value $i$ to rows and querying to check for the existence of $i$ in columns.
    
    In both cases, this uses $O(m)$ updates and queries to our data structure in each OuMv query, and $O(m^2)$ overall.
    Noticing that we create $O(m^2)$ points overall completes the proof.
\end{proof}

\begin{theorem}
   If any of \pointwise{$\setf$}{$+$}{2}, \pointwise{$\setf$}{$\max$}{2}, \pointwise{$+$}{$\max$}{2}, \pointwise{$\max$}{$+$}{2} or \pointwise{$\min$}{$\max$}{2} can be solved in amortised $O(p^{1/2 -\epsilon})$ time per update and amortised $O(p^{1 -\epsilon})$ time per query, for any $\epsilon > 0$, then the \nameref{conj:OMv} is false.
\end{theorem}

%\noindent
%The proof is similar to that of \autoref{omv-gridwise}, and is omitted.
\begin{proof}
    We reduce from OuMv.
    For each $M_{ij} = 1$, create a point $(i, j)$ in our instance of the data structure. 
    
    (\pointwise{$\setf$}{$+$}{2} and \pointwise{$\setf$}{$\max$}{2})
    We maintain that at the beginning of each query, the value of each point is set to 0. 
    We perform the following for each pair of query vectors $u$ and $v$. 
    For each $u_i = 0$, set all points in row $i$ of the data structure to $-1$.
    Do the same for each column $j$ where $v_j = 0$.
    We observe that $u^{T}Mv = 1$ if and only if at least one point in the data structure has a value of 0. 
    We can do a single sum or $\max$ query to determine if this is the case.
    After performing the query, do a single update to reset all the points' values to 0 before receiving the next pair of input vectors.
    
    (\pointwise{$+$}{$\max$}{2})
    We proceed similarly to \pointwise{$\setf$}{$\max$}{2} update, adding $-1$ to all rows $i$ where $u_i = 0$ and columns $j$ where $v_j = 0$.
    In a single query we can determine if a 0 exists in the data structure.
    After performing the query, undo all the updates by adding 1 to the points in each row and column which had an update applied to them.
    (We also note that this proof is similar to that of Henzinger et al. \cite{Henzinger2015} for the hardness of Erickson's problem, which gives an $m \times m$ matrix, and asks, in each operation, to permanently increment all entries in a specified row or column, and report the maximum value in the new matrix).
    
    (\pointwise{$\max$}{$+$}{2})
    We maintain that before processing the $t$th pair of input vectors, the value of each point is at most $t-1$.
    To handle this pair, we perform an update with value $t$ to all rows $i$ where $u_i = 0$ and columns $j$ where $v_j = 0$.
    We observe that $u^{T}Mv = 1$ if and only if at least one point in the data structure has a value less than $t$.
    We can do a single sum query to determine if this is the case.
    
    (\pointwise{$\min$}{$\max$}{2})
    We maintain that at the beginning of each query, the value of each point is at least $-t+1$.
    We perform an update with value $-t$ to all rows $i$ where $u_i = 0$ and columns $j$ where $v_j = 0$.
    We observe that $u^{T}Mv = 1$ if and only if at least one point in the data structure has a value greater than $-t$.
    We can do a single $\max$ query to determine if this is the case.
    
    In all cases, our reduction uses $O(m)$ updates and $O(1)$ queries to our data structure for each OuMv query, and thus $O(m^2)$ updates and $O(m)$ queries overall.
    Finally, noticing that we create $O(m^2)$ points overall gives the desired conditional lower bounds.
\end{proof}

\section{\texorpdfstring{\problemname{$+$}{$\max$}}{2D Grid Range (max, +)} in \texorpdfstring{$\TO(n^{3/2})$}{~O(n\^{}\{3/2)\}} time, and \texorpdfstring{$\TO(n)$}{~O(n)} space}

\label{sec:succinctplusmax}

In this section, we provide a fully-online algorithm for \problemname{$+$}{$\max$} that uses $O(n\log^2 n)$ space.
While our complexity is slower than existing results by Chan \cite{Chan2008} by a polylogarithmic factor, those results require $O(n^{3/2+o(1)})$ space.

We first prove some useful results which are instrumental in the construction of our algorithm.

\begin{lemma}
    \label{lem:1d-f-with-oracle}
    Suppose there is an integer array $A$ of length $n$, containing some initial arbitrary values, and an oracle that supports range maximum queries over the initial $A$ in $O(t)$ time per query.
    There is a data structure which, with access to such an oracle, and with no preprocessing, supports online $+$ and $\max$ updates to $A$, and $\max$ queries over $A$, in worst case $O(t \log n)$ time.
\end{lemma}

\begin{proof}
    Define the class of functions $f_{a, b}(x) = \max(x + a, b)$, which are each monotonically increasing in $x$, and are closed under functional composition.
    As both types of updates can be written in this form, it suffices to only consider updates of this form.
    We use lazy propagation over a segment tree over the ranges of $A$, and at each node, store constants $a$ and $b$, representing $f_{a, b}$ which is the sequential composition of a consecutive range of updates that apply to the range represented by this node.
    At the ``deepest'' nodes we have created, we make calls to the oracle to determine the maximum value in the range, and use this to populate our internal segment tree values.
\end{proof}

Recall that in {\sc Static (Grid) Range} problems, all updates occur before all queries.

\begin{lemma}
    \label{lem:2d-offline-updates-before-queries}
    There exists a data structure which solves {\sc Static} \problemname{$+$}{$\max$} using $O(\numupdates \log^2 \numupdates)$ time in preprocessing and $O(\log^2 \numupdates)$ time per query and $O(\numupdates \log^2 \numupdates)$ space.
\end{lemma}

\begin{proof}
    By applying coordinate compression, it suffices to consider only $O(\numupdates)$ interesting $x$ and $y$ coordinates.

    We will build a kind of 2D segment tree, which is a segment tree over $x$ coordinates.
    Suppose there is a node in the tree representing the range $[x_1, x_2]$.
    For this node, we will build a data structure capable of answering queries of the form: given a range of $y$ coordinates $[y_1, y_2]$, what is the maximum value in $A[x_1, x_2][y_1, y_2]$ (after all updates have been applied)?

    We will build these structures simultaneously, maintaining a sweepline over $x$, representing a column of $A$.
    Break down each range $+$ update into two events: one which adds and another which subtracts a value from a range of $y$ coordinates in the sweepline.
    Using \autoref{lem:1d-f-with-oracle}, the supporting of updates in the sweepline can be done in $O(\log \numupdates)$ time, and range max queries at any point in the sweepline can be supported in $O(\log \numupdates)$ time: the range max value of the initial array is always 0, so the oracle can always return in $O(1)$ time.

    We construct the data structure for the range $[x_1, x_2]$ when the sweepline has processed all events ending at $x_1 - 1$, but before processing any events starting at $x_1$.
    Our data structure will represent an array $B$ such that $B[y] = \max(A[x_1, x_2][y, y])$, and support range queries over $B$.
    Equivalently, we can define $B$ by starting with an array equal to the value of the sweepline after $x = x_1 - 1$, applying all the events in the $[x_1, x_2]$ range, in order, and have $B[y]$ be the historical maximum value of the sweepline in row $y$, at any point in this range.
    This can be computed with a sweepline over the events in the range, and using \autoref{lem:1d-f-with-oracle} together the idea of \cite{Ji2016}, where the oracle is the sweepline at this point in time.
    If there are $\numupdates'$ events in $[x_1, x_2]$, this takes $O(\numupdates' \log \numupdates' \log \numupdates)$ time.
    Each event appears in $O(\log \numupdates)$ ranges $[x_1, x_2]$, so these data structures take $O(\numupdates \log^2 \numupdates)$ time to construct, overall.

    Queries are answered in the usual 2D segment tree manner, which takes $O(\log \numupdates)$ time, for $O(\log^2 \numupdates)$ time overall.
\end{proof}

We are now ready to prove our result.

\oldincmax*

\begin{proof}
    Let $k$ be a positive integer, and split the updates and queries into $n/k$ batches chronologically, such that each contains at most $k$ operations.
    At the beginning of each batch, we build a data structure $D$ according to \autoref{lem:2d-offline-updates-before-queries} from all the updates appearing in previous batches.

    As in \autoref{subsec:sketch}, we split the grid into compressed-columns, based on interesting coordinates in the input so far.
    For each compressed-column $F$, spanning $[x_1, x_2]$, we maintain a data structure which, throughout the operations in this batch, supports queries for the maximum value of a point in $F$, among a given range of rows.
    Such a data structure can be constructed by applying \autoref{lem:1d-f-with-oracle} over the rows of $F$, using $D$ as the oracle.
    When we receive new coordinates, we split compressed-columns in the same way as in \autoref{subsec:formal}, reapplying updates.
    Each such query takes $\TO(1)$ time, so together, queries take $\TO(k^2)$ time over the batch.

    Factoring in the construction time of $D$, the overall time complexity of this algorithm is $\TO((n/k)(n + k^2))$.
    Setting $k = \sqrt{n}$ gives the desired complexity.
\end{proof}

\section{Open Problems}

\label{sec:open}

We have shown that \problemname{$\{\min, \setf, \max\}$}{$\max$} and \problemname{$+$}{$\max$} can both be solved in truly subquadratic time, and found $\Omega(n^{2-o(1)})$ time conditional lower bounds for \problemname{$\{+, \min\}$}{$\max$}.
We also observe that \problemname{$\{+, \setf\}$}{$\max$} reduces to \problemname{$\{+, \max\}$}{$\max$}, since $\setf_c = \max_c \circ +_{-\infty}$.
Hence, the remaining maximum query variants in \probs{} are each at least as hard as \problemname{$\{+, \setf\}$}{$\max$}.

\begin{openproblem}
    Can {\sc (Offline)} \problemname{$\{+, \setf\}$}{$\max$} be solved in truly subquadratic time?
\end{openproblem} 

Among variants supporting sum queries, we gave $\Omega(n^{2-o(1)})$ time conditional lower bounds for \problemname{$\{+, \max\}$}{$+$}.
Using the identity $\setf_c = \max_c \circ +_{-\infty}$ once again, one can see that this is at least as hard as \problemname{$\{+, \setf\}$}{$+$}, which we solved in $O(n^{1.989})$ time, using \autoref{general-static-trellis}.
Another problem easier than \problemname{$\{+, \max\}$}{$+$} is simply \problemname{$\max$}{$+$}, which does not support $+$ updates.
This is also the easiest among the remaining sum query variants in \probs.

We make several comments regarding the hardness of \problemname{$\max$}{$+$}.
First, we observe that it is also at least as hard as its $\setf$ ``counterpart'', since any instance of {\sc Grid \range{$\setf$}{$+$}} can be simulated with two instances of {\sc Grid \range{$\max$}{$+$}}.

\begin{lemma}
    {\sc Grid \range{$\setf$}{$+$}} can be solved in the same time as {\sc Grid \range{$\max$}{$+$}}.
\end{lemma}

We solved \problemname{$\setf$}{$+$} in $O(n^{1.954})$ time using \autoref{general-static-trellis}.
However, it is easy to see that there is no solution to {\sc 1D Partitioned Range $(+, \max)$}, which is required as a precondition of \autoref{general-static-trellis}: it is simply not enough to know the sum of a range of points.
One might instead determine for each overlay region $O$, query range $R$ and $\max_c$ update: the number of points in $O \cap R$ with value at most $c$ and the sum of points in $O \cap R$ with value greater than $c$.
This requires $O(k^3)$ values returned per batch, limiting precomputation to $O(n^{3/2-\epsilon})$ time, for some $\epsilon > 0$, if a truly subquadratic time algorithm overall is desired.
This cannot be achieved with a direct application of \autoref{general-static-trellis}, since there are $O(n^{3/2})$ updates across the t-regions. % \ar{Should this be: since there are $n^{3/2}$ updates across the t-regions (since there are O(n) t-regions)}

\begin{openproblem}
    Can {\sc (Offline)} \problemname{$\max$}{$+$} be solved in truly subquadratic time?
\end{openproblem}

Finally, our $\Omega(n^{3/2-o(1)})$ conditional lower bounds do not match the upper bounds we gave for \problemname{$\{\setf, +\}$}{$+$}, \problemname{$\setf$}{$+$} and \problemname{$\{\min, \setf, \max\}$}{$\max$}.
We ask if the gap can be closed for these problems to see if there exists an in-between complexity class of {\sc 2D Grid Range} problems.
In particular, this would be resolved in the affirmative if \autoref{general-static-trellis} is a tight reduction for any of these problems.

\begin{openproblem}
    Are there any {\sc 2D Grid Range} problems solvable in $O(n^{2-\epsilon})$ time, for some $\epsilon > 0$, but require $\Omega(n^{3/2-o(1)})$ time?
\end{openproblem}

We have studied just a small subset of {\sc Range} and {\sc Grid Range} problems in this work.
Additional update or query operations, such as addition modulo a prime, can also be considered.
Many existing variants of range searching (see \cite{Agarwal2004}) can also be adapted to these problem classes.
In particular, we have not investigated problems which deal with data points that have a ``colour'' or ``category'', and ask for the number of distinct colours in a range.
These may be of particular interest, as $\setf$ updates could be used to facilitate changing the colour of several data points at the same time.

\bibliography{references}

\begin{thebibliography}{10}

\bibitem{Abboud2015}
Amir Abboud, Virginia~Vassilevska Williams, and Huacheng Yu.
\newblock Matching triangles and basing hardness on an extremely popular
  conjecture.
\newblock In Rocco~A. Servedio and Ronitt Rubinfeld, editors, {\em Proceedings
  of the Forty-Seventh Annual {ACM} on Symposium on Theory of Computing, {STOC}
  2015, Portland, OR, USA, June 14-17, 2015}, pages 41--50. {ACM}, 2015.
\newblock \href {http://dx.doi.org/10.1145/2746539.2746594}
  {\path{doi:10.1145/2746539.2746594}}.

\bibitem{Afshani2019}
Peyman Afshani.
\newblock A new lower bound for semigroup orthogonal range searching.
\newblock In Gill Barequet and Yusu Wang, editors, {\em 35th International
  Symposium on Computational Geometry, SoCG 2019, June 18-21, 2019, Portland,
  Oregon, {USA}}, volume 129 of {\em LIPIcs}, pages 3:1--3:14. Schloss Dagstuhl
  - Leibniz-Zentrum f{\"{u}}r Informatik, 2019.
\newblock \href {http://dx.doi.org/10.4230/LIPIcs.SoCG.2019.3}
  {\path{doi:10.4230/LIPIcs.SoCG.2019.3}}.

\bibitem{Afshani2009}
Peyman Afshani, Lars Arge, and Kasper~Dalgaard Larsen.
\newblock Orthogonal range reporting in three and higher dimensions.
\newblock In {\em 50th Annual {IEEE} Symposium on Foundations of Computer
  Science, {FOCS} 2009, October 25-27, 2009, Atlanta, Georgia, {USA}}, pages
  149--158. {IEEE} Computer Society, 2009.
\newblock \href {http://dx.doi.org/10.1109/FOCS.2009.58}
  {\path{doi:10.1109/FOCS.2009.58}}.

\bibitem{Afshani2010}
Peyman Afshani, Lars Arge, and Kasper~Dalgaard Larsen.
\newblock Orthogonal range reporting: query lower bounds, optimal structures in
  3-d, and higher-dimensional improvements.
\newblock In David~G. Kirkpatrick and Joseph S.~B. Mitchell, editors, {\em
  Proceedings of the 26th {ACM} Symposium on Computational Geometry, Snowbird,
  Utah, USA, June 13-16, 2010}, pages 240--246. {ACM}, 2010.
\newblock \href {http://dx.doi.org/10.1145/1810959.1811001}
  {\path{doi:10.1145/1810959.1811001}}.

\bibitem{Agarwal2004}
Pankaj~K. Agarwal.
\newblock Range searching.
\newblock In Jacob~E. Goodman and Joseph O'Rourke, editors, {\em Handbook of
  Discrete and Computational Geometry, Second Edition}, pages 809--837. Chapman
  and Hall/CRC, 2004.
\newblock \href {http://dx.doi.org/10.1201/9781420035315.ch36}
  {\path{doi:10.1201/9781420035315.ch36}}.

\bibitem{Alman2020}
Josh Alman and Virginia~Vassilevska Williams.
\newblock A refined laser method and faster matrix multiplication.
\newblock {\em CoRR}, abs/2010.05846, 2020.
\newblock URL: \url{https://arxiv.org/abs/2010.05846}, \href
  {http://arxiv.org/abs/2010.05846} {\path{arXiv:2010.05846}}.

\bibitem{Backurs2016}
Arturs Backurs, Nishanth Dikkala, and Christos Tzamos.
\newblock Tight hardness results for maximum weight rectangles.
\newblock In Ioannis Chatzigiannakis, Michael Mitzenmacher, Yuval Rabani, and
  Davide Sangiorgi, editors, {\em 43rd International Colloquium on Automata,
  Languages, and Programming, {ICALP} 2016, July 11-15, 2016, Rome, Italy},
  volume~55 of {\em LIPIcs}, pages 81:1--81:13. Schloss Dagstuhl -
  Leibniz-Zentrum f{\"{u}}r Informatik, 2016.
\newblock \href {http://dx.doi.org/10.4230/LIPIcs.ICALP.2016.81}
  {\path{doi:10.4230/LIPIcs.ICALP.2016.81}}.

\bibitem{Bentley1975}
Jon~Louis Bentley.
\newblock Multidimensional binary search trees used for associative searching.
\newblock {\em Commun. {ACM}}, 18(9):509--517, 1975.
\newblock \href {http://dx.doi.org/10.1145/361002.361007}
  {\path{doi:10.1145/361002.361007}}.

\bibitem{Bentley1977}
Jon~Louis Bentley.
\newblock Solutions to klee’s rectangle problems.
\newblock {\em Unpublished manuscript}, pages 282--300, 1977.

\bibitem{Chan2008}
Timothy~M. Chan.
\newblock A (slightly) faster algorithm for klee's measure problem.
\newblock In Monique Teillaud, editor, {\em Proceedings of the 24th {ACM}
  Symposium on Computational Geometry, College Park, MD, USA, June 9-11, 2008},
  pages 94--100. {ACM}, 2008.
\newblock \href {http://dx.doi.org/10.1145/1377676.1377693}
  {\path{doi:10.1145/1377676.1377693}}.

\bibitem{Chan2013}
Timothy~M. Chan.
\newblock Klee's measure problem made easy.
\newblock In {\em 54th Annual {IEEE} Symposium on Foundations of Computer
  Science, {FOCS} 2013, 26-29 October, 2013, Berkeley, CA, {USA}}, pages
  410--419. {IEEE} Computer Society, 2013.
\newblock \href {http://dx.doi.org/10.1109/FOCS.2013.51}
  {\path{doi:10.1109/FOCS.2013.51}}.

\bibitem{Chan2017a}
Timothy~M. Chan.
\newblock Orthogonal range searching in moderate dimensions: k-d trees and
  range trees strike back.
\newblock In Boris Aronov and Matthew~J. Katz, editors, {\em 33rd International
  Symposium on Computational Geometry, SoCG 2017, July 4-7, 2017, Brisbane,
  Australia}, volume~77 of {\em LIPIcs}, pages 27:1--27:15. Schloss Dagstuhl -
  Leibniz-Zentrum f{\"{u}}r Informatik, 2017.
\newblock \href {http://dx.doi.org/10.4230/LIPIcs.SoCG.2017.27}
  {\path{doi:10.4230/LIPIcs.SoCG.2017.27}}.

\bibitem{Chan2011}
Timothy~M. Chan, Kasper~Green Larsen, and Mihai P{\v a}tra{\c s}cu.
\newblock Orthogonal range searching on the ram, revisited.
\newblock In Ferran Hurtado and Marc~J. van Kreveld, editors, {\em Proceedings
  of the 27th {ACM} Symposium on Computational Geometry, Paris, France, June
  13-15, 2011}, pages 1--10. {ACM}, 2011.
\newblock \href {http://dx.doi.org/10.1145/1998196.1998198}
  {\path{doi:10.1145/1998196.1998198}}.

\bibitem{Chan2019}
Timothy~M. Chan, Yakov Nekrich, and Michiel H.~M. Smid.
\newblock Orthogonal range reporting and rectangle stabbing for fat rectangles.
\newblock In Zachary Friggstad, J{\"{o}}rg{-}R{\"{u}}diger Sack, and
  Mohammad~R. Salavatipour, editors, {\em Algorithms and Data Structures - 16th
  International Symposium, {WADS} 2019, Edmonton, AB, Canada, August 5-7, 2019,
  Proceedings}, volume 11646 of {\em Lecture Notes in Computer Science}, pages
  283--295. Springer, 2019.
\newblock URL: \url{https://doi.org/10.1007/978-3-030-24766-9_21}, \href
  {http://dx.doi.org/10.1007/978-3-030-24766-9\_21}
  {\path{doi:10.1007/978-3-030-24766-9\_21}}.

\bibitem{Chan2017}
Timothy~M. Chan and Konstantinos Tsakalidis.
\newblock Dynamic orthogonal range searching on the ram, revisited.
\newblock In Boris Aronov and Matthew~J. Katz, editors, {\em 33rd International
  Symposium on Computational Geometry, SoCG 2017, July 4-7, 2017, Brisbane,
  Australia}, volume~77 of {\em LIPIcs}, pages 28:1--28:13. Schloss Dagstuhl -
  Leibniz-Zentrum f{\"{u}}r Informatik, 2017.
\newblock \href {http://dx.doi.org/10.4230/LIPIcs.SoCG.2017.28}
  {\path{doi:10.4230/LIPIcs.SoCG.2017.28}}.

\bibitem{Chazelle1988}
Bernard Chazelle.
\newblock A functional approach to data structures and its use in
  multidimensional searching.
\newblock {\em {SIAM} J. Comput.}, 17(3):427--462, 1988.
\newblock \href {http://dx.doi.org/10.1137/0217026}
  {\path{doi:10.1137/0217026}}.

\bibitem{Driscoll1986}
James~R. Driscoll, Neil Sarnak, Daniel~Dominic Sleator, and Robert~Endre
  Tarjan.
\newblock Making data structures persistent.
\newblock In Juris Hartmanis, editor, {\em Proceedings of the 18th Annual {ACM}
  Symposium on Theory of Computing, May 28-30, 1986, Berkeley, California,
  {USA}}, pages 109--121. {ACM}, 1986.
\newblock \href {http://dx.doi.org/10.1145/12130.12142}
  {\path{doi:10.1145/12130.12142}}.

\bibitem{Duraj2020}
Lech Duraj, Krzysztof Kleiner, Adam Polak, and Virginia~Vassilevska Williams.
\newblock Equivalences between triangle and range query problems.
\newblock In Shuchi Chawla, editor, {\em Proceedings of the 2020 {ACM-SIAM}
  Symposium on Discrete Algorithms, {SODA} 2020, Salt Lake City, UT, USA,
  January 5-8, 2020}, pages 30--47. {SIAM}, 2020.
\newblock \href {http://dx.doi.org/10.1137/1.9781611975994.3}
  {\path{doi:10.1137/1.9781611975994.3}}.

\bibitem{Farzan2012}
Arash Farzan, J.~Ian Munro, and Rajeev Raman.
\newblock Succinct indices for range queries with applications to orthogonal
  range maxima.
\newblock In Artur Czumaj, Kurt Mehlhorn, Andrew~M. Pitts, and Roger
  Wattenhofer, editors, {\em Automata, Languages, and Programming - 39th
  International Colloquium, {ICALP} 2012, Warwick, UK, July 9-13, 2012,
  Proceedings, Part {I}}, volume 7391 of {\em Lecture Notes in Computer
  Science}, pages 327--338. Springer, 2012.
\newblock URL: \url{https://doi.org/10.1007/978-3-642-31594-7_28}, \href
  {http://dx.doi.org/10.1007/978-3-642-31594-7\_28}
  {\path{doi:10.1007/978-3-642-31594-7\_28}}.

\bibitem{Flum2004}
J{\"{o}}rg Flum and Martin Grohe.
\newblock The parameterized complexity of counting problems.
\newblock {\em {SIAM} J. Comput.}, 33(4):892--922, 2004.
\newblock \href {http://dx.doi.org/10.1137/S0097539703427203}
  {\path{doi:10.1137/S0097539703427203}}.

\bibitem{Gray1996}
Jim Gray, Adam Bosworth, Andrew Layman, and Hamid Pirahesh.
\newblock Data cube: {A} relational aggregation operator generalizing group-by,
  cross-tab, and sub-total.
\newblock In Stanley Y.~W. Su, editor, {\em Proceedings of the Twelfth
  International Conference on Data Engineering, February 26 - March 1, 1996,
  New Orleans, Louisiana, {USA}}, pages 152--159. {IEEE} Computer Society,
  1996.
\newblock \href {http://dx.doi.org/10.1109/ICDE.1996.492099}
  {\path{doi:10.1109/ICDE.1996.492099}}.

\bibitem{He2014}
Meng He and J.~Ian Munro.
\newblock Space efficient data structures for dynamic orthogonal range
  counting.
\newblock {\em Comput. Geom.}, 47(2):268--281, 2014.
\newblock \href {http://dx.doi.org/10.1016/j.comgeo.2013.08.007}
  {\path{doi:10.1016/j.comgeo.2013.08.007}}.

\bibitem{Henzinger2015}
Monika Henzinger, Sebastian Krinninger, Danupon Nanongkai, and Thatchaphol
  Saranurak.
\newblock Unifying and strengthening hardness for dynamic problems via the
  online matrix-vector multiplication conjecture.
\newblock In Rocco~A. Servedio and Ronitt Rubinfeld, editors, {\em Proceedings
  of the Forty-Seventh Annual {ACM} on Symposium on Theory of Computing, {STOC}
  2015, Portland, OR, USA, June 14-17, 2015}, pages 21--30. {ACM}, 2015.
\newblock \href {http://dx.doi.org/10.1145/2746539.2746609}
  {\path{doi:10.1145/2746539.2746609}}.

\bibitem{Ibtehaz2018}
Nabil Ibtehaz, M.~Kaykobad, and M.~Sohel Rahman.
\newblock Multidimensional segment trees can do range queries and updates in
  logarithmic time.
\newblock {\em CoRR}, abs/1811.01226, 2018.
\newblock URL: \url{http://arxiv.org/abs/1811.01226}, \href
  {http://arxiv.org/abs/1811.01226} {\path{arXiv:1811.01226}}.

\bibitem{Ji2016}
Ruyi Ji.
\newblock Interval maximum value operation and historical maximum value
  problem.
\newblock {\em Informatics Olympiad China National Team Candidates Essay
  Collection}, 2016.
\newblock Article written in Chinese. The author (Ji) has written a blog post
  in English, outlining the main techniques:
  https://codeforces.com/blog/entry/57319.

\bibitem{Korhonen2019}
Tuukka Korhonen.
\newblock On multidimensional range queries.
\newblock Technical report, University of Helsinki, 2019.

\bibitem{Lee1977}
D.~T. Lee and C.~K. Wong.
\newblock Worst-case analysis for region and partial region searches in
  multidimensional binary search trees and balanced quad trees.
\newblock {\em Acta Informatica}, 9:23--29, 1977.
\newblock \href {http://dx.doi.org/10.1007/BF00263763}
  {\path{doi:10.1007/BF00263763}}.

\bibitem{Lueker1978}
George~S. Lueker.
\newblock A data structure for orthogonal range queries.
\newblock In {\em 19th Annual Symposium on Foundations of Computer Science, Ann
  Arbor, Michigan, USA, 16-18 October 1978}, pages 28--34. {IEEE} Computer
  Society, 1978.
\newblock \href {http://dx.doi.org/10.1109/SFCS.1978.1}
  {\path{doi:10.1109/SFCS.1978.1}}.

\bibitem{Nekrich2020}
Yakov Nekrich.
\newblock New data structures for orthogonal range reporting and range minima
  queries.
\newblock {\em CoRR}, abs/2007.11094, 2020.
\newblock URL: \url{https://arxiv.org/abs/2007.11094}, \href
  {http://arxiv.org/abs/2007.11094} {\path{arXiv:2007.11094}}.

\bibitem{Okajima2015}
Yuzuru Okajima and Kouichi Maruyama.
\newblock Faster linear-space orthogonal range searching in arbitrary
  dimensions.
\newblock In Ulrik Brandes and David Eppstein, editors, {\em Proceedings of the
  Seventeenth Workshop on Algorithm Engineering and Experiments, {ALENEX} 2015,
  San Diego, CA, USA, January 5, 2015}, pages 82--93. {SIAM}, 2015.
\newblock \href {http://dx.doi.org/10.1137/1.9781611973754.8}
  {\path{doi:10.1137/1.9781611973754.8}}.

\bibitem{Overmars1991}
Mark~H. Overmars and Chee{-}Keng Yap.
\newblock New upper bounds in klee's measure problem.
\newblock {\em {SIAM} J. Comput.}, 20(6):1034--1045, 1991.
\newblock \href {http://dx.doi.org/10.1137/0220065}
  {\path{doi:10.1137/0220065}}.

\bibitem{Poon2003}
Chung~Keung Poon.
\newblock Dynamic orthogonal range queries in {OLAP}.
\newblock {\em Theor. Comput. Sci.}, 296(3):487--510, 2003.
\newblock \href {http://dx.doi.org/10.1016/S0304-3975(02)00741-7}
  {\path{doi:10.1016/S0304-3975(02)00741-7}}.

\bibitem{Williams2010}
Virginia {Vassilevska Williams} and Ryan Williams.
\newblock Subcubic equivalences between path, matrix and triangle problems.
\newblock In {\em 51th Annual {IEEE} Symposium on Foundations of Computer
  Science, {FOCS} 2010, October 23-26, 2010, Las Vegas, Nevada, {USA}}, pages
  645--654. {IEEE} Computer Society, 2010.
\newblock \href {http://dx.doi.org/10.1109/FOCS.2010.67}
  {\path{doi:10.1109/FOCS.2010.67}}.

\bibitem{Willard1985}
Dan~E. Willard and George~S. Lueker.
\newblock Adding range restriction capability to dynamic data structures.
\newblock {\em J. {ACM}}, 32(3):597--617, 1985.
\newblock \href {http://dx.doi.org/10.1145/3828.3839}
  {\path{doi:10.1145/3828.3839}}.

\bibitem{Williams2005}
Ryan Williams.
\newblock A new algorithm for optimal 2-constraint satisfaction and its
  implications.
\newblock {\em Theor. Comput. Sci.}, 348(2-3):357--365, 2005.
\newblock \href {http://dx.doi.org/10.1016/j.tcs.2005.09.023}
  {\path{doi:10.1016/j.tcs.2005.09.023}}.

\bibitem{Williams2015a}
Virginia~Vassilevska Williams.
\newblock Hardness of easy problems: Basing hardness on popular conjectures
  such as the strong exponential time hypothesis (invited talk).
\newblock In Thore Husfeldt and Iyad~A. Kanj, editors, {\em 10th International
  Symposium on Parameterized and Exact Computation, {IPEC} 2015, September
  16-18, 2015, Patras, Greece}, volume~43 of {\em LIPIcs}, pages 17--29.
  Schloss Dagstuhl - Leibniz-Zentrum f{\"{u}}r Informatik, 2015.
\newblock \href {http://dx.doi.org/10.4230/LIPIcs.IPEC.2015.17}
  {\path{doi:10.4230/LIPIcs.IPEC.2015.17}}.

\bibitem{Yuster2004}
Raphael Yuster and Uri Zwick.
\newblock Detecting short directed cycles using rectangular matrix
  multiplication and dynamic programming.
\newblock In J.~Ian Munro, editor, {\em Proceedings of the Fifteenth Annual
  {ACM-SIAM} Symposium on Discrete Algorithms, {SODA} 2004, New Orleans,
  Louisiana, USA, January 11-14, 2004}, pages 254--260. {SIAM}, 2004.
\newblock URL: \url{http://dl.acm.org/citation.cfm?id=982792.982828}.

\end{thebibliography}

\end{document}